\documentclass[
onefignum,onetabnum]{siamart171218}

\usepackage{amsmath}
\usepackage{amssymb}
\usepackage{amsfonts}
\usepackage{graphicx}
\usepackage{lipsum}
\usepackage{epstopdf}
\usepackage{algorithmic}
\usepackage{datetime}
\ifpdf
  \DeclareGraphicsExtensions{.eps,.pdf,.png,.jpg}
\else
  \DeclareGraphicsExtensions{.eps}
\fi

\newsiamremark{remark}{Remark}
\newsiamremark{hypothesis}{Hypothesis}
\crefname{hypothesis}{Hypothesis}{Hypotheses}
\newsiamthm{claim}{Claim}

\headers{Decoherence 
in Quasilinear Quantum Systems}{I.G.Vladimirov, I.R.Petersen}

\def\vect{\mathrm{vec}}

\def\sgn{\mathrm{sgn}}

\def\<{\leqslant}           
\def\>{\geqslant}           

\def\d{\partial}
\def\wh{\widehat}
\def\wt{\widetilde}

\def\Re{\mathrm{Re}}   
\def\Im{\mathrm{Im}}   

\def\mR{\mathbb{R}}    
\def\mC{\mathbb{C}}    

\def\Tr{\mathrm{Tr}}       
\def\rT{\mathrm{T}}        
\def\rF{\mathrm{F}}
\def\diam{\diamond}

\def\bS{\mathbf{S}}

\def\bE{\mathbf{E}}    

\def\bra{\langle}
\def\ket{\rangle}

\def\re{\mathrm{e}}        
\def\rd{\mathrm{d}}        

\def\bJ{\mathbf{J}}

\def\br{\mathbf{r}}
\def\x{\times}
\def\ox{\otimes}
\def\op{\oplus}

\def\fB{\mathfrak{B}}

\def\fF{\mathfrak{F}}

\def\fH{\mathfrak{H}}

\def\bzero{\mathbf{0}}
\def\sb{\mathsf{b}}
\def\sM{\mathsf{M}}
\def\sL{\mathsf{L}}
\def\sGamma{\mathsf{\Gamma}}
\def\sN{\mathsf{N}}

\def\cF{\mathcal{F}}

\def\cD{\mathcal{D}}

\def\sB{\mathsf{B}}

\def\sP{\mathsf{P}}
\def\sA{\mathsf{A}}

\def\cG{\mathcal{G}}
\def\cI{\mathcal{I}}

\def\cA{\mathcal{A}}
\def\cB{\mathcal{B}}
\def\cE{\mathcal{E}}

\def\im{\mathrm{im\,}}

\def\mH{\mathbb{H}}

\def\eps{\epsilon}
\def\Ups{\Upsilon}
\def\ups{\upsilon}

\def\rprod{\mathop{\overrightarrow{\prod}}}
\def\lexp{\mathop{\overleftarrow{\exp}}}

\DeclareMathAlphabet{\mathbfit}{OML}{cmm}{b}{it}

\title{Measuring    Decoherence by Commutation Relations Decay for Quasilinear Quantum Stochastic Systems\thanks{This work is supported by the Australian Research Council  grant DP210101938.}}

\author{
Igor G. Vladimirov$^{\dagger}$, \quad Ian R. Petersen\thanks{School of Engineering, College of Engineering and Computer Science, Australian National University, Canberra, ACT 2601, Australia,
  \email{igor.g.vladimirov@gmail.com, i.r.petersen@gmail.com}.}}

\usepackage{amsopn}
\DeclareMathOperator{\diag}{diag}
\def\blockdiag{\mathop\mathrm{blockdiag}}

\begin{document}

\maketitle

\begin{abstract}
This paper considers a class of open quantum systems with an algebraic structure of dynamic  variables, including the Pauli matrices for finite-level systems as a particular case. The Hamiltonian and the operators of coupling of the system to the external bosonic fields depend linearly on the system variables. The fields are represented by quantum Wiener processes which drive the system dynamics in the form of a quasilinear Hudson-Parthasarathy quantum stochastic differential equation whose drift vector and dispersion matrix are affine and linear functions of the system variables. This quasilinearity leads to a tractable evolution of the two-point  commutator matrix of the system variables (and their multi-point mixed moments in the case of vacuum input fields) involving time-ordered operator exponentials. The resulting exponential decay in the two-point commutation relations is a manifestation of quantum decoherence, caused by the dissipative system-field interaction and making the system lose specific unitary dynamics features which it would have in isolation from the environment.
We quantify the decoherence   in terms of the rate of the commutation relations decay  and apply system theoretic and matrix analytic techniques, such as  algebraic Lyapunov inequalities and spectrum perturbation results, to the study of the asymptotic behaviour of the related Lyapunov exponents in the presence of a small scaling parameter in the system-field coupling.
These findings are illustrated for finite-level quantum systems (and their interconnections  through a direct energy coupling) with multichannel external fields and the Pauli matrices as internal variables.


\end{abstract}



\section{Introduction}

Quantum mechanics, which extends classical (deterministic or stochastic)  dynamics in order to describe physical objects at atomic scales, models the quantum systems and their evolution in terms of linear operators on a Hilbert space. Normally, these operators result from the quantisation of classical real-valued variables (for example, positions, momenta and angular momenta from Hamiltonian mechanics and rigid body dynamics) and are self-adjoint, although non-Hermitian operators (such as the annihilation and creation operators \cite{S_1994}) are also employed. Not only the dynamic  variables, but also the statistical properties of quantum systems are described using operator-valued quantities. The latter include density operators on the underlying  Hilbert space,  which specify quantum states as a quantum probabilistic analogue  \cite{H_2001,M_1995} of scalar-valued  classical probability distributions.

The noncommutativity of quantum dynamic  variables and quantum states, coming from their  operator-valued nature, plays a part in the model of quantum measurement as an interaction of a quantum system with a classical measuring device  which modifies the quantum state in such a way that noncommuting observables cannot be measured simultaneously. As opposed to classical random variables, such quantum observables do not have a classical  joint distribution (factorisable into marginal and conditional probability distributions), and moreover,  the corresponding joint quantum state can acquire nonclassical properties such as nonseparability (or entanglement).
In addition to the interplay  between the noncommutativity and measurements as a specific feature  of quantum mechanics, an isolated quantum system (not interacting with the environment)   is also modelled differently from classical  systems ---  in terms of a unitary similarity transformation group generated by the system Hamiltonian (the energy operator). Due to the unitary (and hence, reversible) structure of the isolated quantum dynamics, the initial conditions of such a system have a nonvanishing influence on its subsequent evolution, which manifests itself in conservation laws, including the preservation of the  Hamiltonian and one-point canonical commutation relations (CCRs).

The above features provide nonclassical resources which are exploited in quantum information processing technologies,  such as quantum error correction \cite{NC_2000}. However, the unitary quantum evolution is never ideally realised in practice, since any quantum system cannot be perfectly isolated from and unavoidably interacts with its surroundings (which include other quantum or classical systems and external fields). As  a result, the internal dynamics, which the quantum system would have in isolation,  are ``diluted'' with the forced  response to the environment whose variables (as operators on different spaces) commute with (and hence, behave like classical ones with respect to) those of the system. Over the course of this interaction, the open system gradually loses its individual quantum features in terms of the commutation structure and statistical properties of the system variables, which is usually referred to as quantum decoherence \cite{BP_2006,GZ_2004}.

The open quantum dynamics, along with concomitant decoherence effects, can be modelled using the Hudson-Parthasarathy calculus \cite{HP_1984,P_1992} of quantum stochastic differential equations (QSDEs) driven by noncommutative Wiener processes on a symmetric Fock space \cite{PS_1972}. Such processes represent the external bosonic fields and lead to more complicated quantum Ito processes associated with the internal variables and output fields of open quantum systems. The drift vector and the dispersion matrix of the QSDE are specified by the energetics of the system and its interaction with the environment, which is   captured in the system Hamiltonian and the operators of coupling of the system to the input fields. The structure of the resulting  QSDE is affected by the dependence of the Hamiltonian and the coupling operators on the system variables and the algebraic properties of the latter, including the commutation relations.

In the context of quantum computing and quantum information, the study of decoherence phenomena is particularly relevant for finite-level (multiqubit) systems \cite[Chapter III, Section 8]{NC_2000} which employ the Pauli matrices \cite{S_1994} (special Hermitian $(2\x2)$-matrices originating from the theory of quantum angular momentum, including the electron spin)   and their higher-dimensional extensions, such as the Gell-Mann matrices (see, for example, \cite{EMPUJ_2016} and references therein). The sets of such quantum variables are algebraically closed in the sense that their products are linear combinations of  the same set of variables and the identity operator, thus allowing any function (for example, a polynomial) of such variables to be reduced to an affine function and also leading to CCRs for the variables.  In the presence of the algebraic structure, the system Hamiltonian and the system-field coupling operators are, without loss of generality,  linear and affine functions of the system variables, which results in QSDEs whose drift vector depends affinely and the dispersion matrix depends linearly on the system variables, as in classical SDEs with state-dependent noise \cite{W_1967}.

Although these quasilinear QSDEs \cite{EMPUJ_2016,VP_2022} (with a bilinear dependence on the system and input field variables) are different from linear QSDEs (with a constant dispersion matrix) for open quantum harmonic oscillators (OQHOs) as principal system models in linear quantum control theory \cite{B_1983,DP_2010,NY_2017,P_2017}, their solutions are expressed through fundamental solutions in the form of time-ordered operator exponentials \cite{H_1996}.  These time-ordered exponentials, which  relate the system variables at different instants and participate in the two-point CCRs between them, are more complicated than the usual matrix exponentials \cite{H_2008} (describing the fundamental solutions of linear ODEs with constant coefficients). Nevertheless, as shown in \cite{VP_2022}, they admit a closed-form quantum averaging at least for the vacuum quantum state \cite{P_1992} of the external fields. This makes the mean values and higher-order multipoint moments of the system variables (including their two-point covariances) effectively computable, with the reducibility  of nonlinear functions of the system variables to affine functions due to the algebraic structure also  playing its role here. The statistical characteristics of the invariant quantum state of the system  (provided the dynamics matrix of the quasilinear QSDE is Hurwitz), captured in  its quasi-characteristic function \cite{CH_1971}, are also amenable to closed-form computation along with  the rate of the exponentially fast convergence to this state at the level of moments. Similarly to \cite{Y_2012,Y_2009}, such convergence is useful for achieving steady-state regimes in quantum systems (subject to external vacuum noise) and benefits from the dissipative nature of the system-field interaction. On the other hand, the enhanced dissipation, which accelerates the invariant state generation, contributes to decoherence effects in comparison with the isolated quantum dynamics.

The present paper is concerned with decoherence measures for the  class of open quantum stochastic systems whose dynamic variables have an algebraic structure (thus leading to a linear Hamiltonian and affine system-field coupling operators) and are  governed by quasilinear QSDEs mentioned above.  Building on the results of \cite{EMPUJ_2016,VP_2022}, we exploit the tractability of moment dynamics, which comes from this quasilinearity, and study the rate of the  exponential decay in the two-point commutator matrix for the system variables.
This is carried out for the case of vacuum input fields in terms of the Lyapunov exponents for the mean values and second moments of the time-ordered exponentials, which, as the fundamental solutions of the QSDEs,   play an important role in the two-point CCRs. Similarly to \cite{VP_2022_decoh}, a decoherence time constant, associated with   the CCR decay rate,  is compared with the periods of oscillatory modes which the system would have in isolation from the external fields, thus quantifying the decoherence caused by the dissipative system-field interaction. We apply system theoretic and matrix analytic techniques, such as  algebraic Lyapunov inequalities and spectrum perturbation results \cite{M_1985}, in order to obtain upper bounds for the decoherence time and related Lyapunov exponents and to investigate their  asymptotic behaviour in the presence of a small scaling parameter in the system-field coupling with a fixed coupling shape. This analysis employs conditions of stability of the system for all sufficiently small values of the coupling strength, which also allows the asymptotic behaviour of the invariant system state to be investigated in the weak-coupling formulation.
These results are illustrated for finite-level quantum systems with multichannel external fields and the Pauli matrices as internal variables. We also apply the findings of the paper   to a decoherence  control setting for an interconnection of two quasilinear quantum systems (interpreted as a plant and a controller) with a direct energy coupling \cite{ZJ_2011a}.

We also mention that there are various scenarios of quantum decoherence, coming, for example, from quantum measurements and thermal noise and also those concerned with different classes of quantum systems (see \cite[Section 4.4.1]{BP_2006} and references therein, including \cite{CL_1985,U_1995}). However, this paper focuses on one of manifestations of decoherence through the two-point CCR decay (which corresponds to vacuum decoherence \cite[Section 4.4.1]{BP_2006}), extending this circle of ideas from OQHOs \cite{VP_2022_decoh}   (with qualitatively different system variables in the form of unbounded position  and momentum operators) to finite-level  quantum  systems.

The paper is organised as follows.
Section~\ref{sec:var} specifies a class of quantum dynamic variables with an algebraic structure and describes their operator theoretic properties (including CCRs and boundedness), exemplified by the Pauli matrices.
Section~\ref{sec:QSS} provides  a quasilinear QSDE  for open quantum systems with such variables, a linear Hamiltonian and affine coupling operators, discusses energy balance relations and convergence to the invariant state and  reviews the representation of the quantum trajectories in terms of time-ordered operator   exponentials.
Section~\ref{sec:decay} studies the averaging of these exponentials in order to obtain  Lyapunov exponents for the decay rate in the two-point CCRs of the system variables.
Section~\ref{sec:iso} discusses the oscillatory modes of the quantum system in isolation from the environment, which  provide reference time scales for comparison with the dissipative open quantum dynamics,  and uses algebraic Lyapunov  inequalities to establish upper bounds for the decoherence time associated with the two-point CCR decay.
Section~\ref{sec:asy} employs matrix spectrum perturbation techniques and develops stability conditions for the system and asymptotic estimates for the decoherence  time in a setting where the system-field coupling involves a small strength parameter along with a given coupling shape.
Section~\ref{sec:inv} studies the asymptotic behaviour of the invariant state of the system in the weak-coupling formulation.
Section~\ref{sec:Pauli} computes in closed-form several quantities, relevant to the decay rates in the open quantum system and oscillatory modes of its isolated version in the case of the Pauli matrices as system variables.
Section~\ref{sec:connect} applies the results of the previous sections  to an interconnection of two open quantum systems which interact with the external fields and are directly coupled to each other.
Section~\ref{sec:conc} provides concluding remarks.

\section{Algebraic structure of quantum variables}
\label{sec:var}


Following \cite{VP_2022}, we consider an  open  quantum stochastic system with $n$ dynamic variables $X_1(t), \ldots, X_n(t)$, organised as time-varying self-adjoint operators on a complex separable Hilbert space $\fH$ with the algebraic structure
\begin{equation}
\label{XXX}
    \Xi_{jk}(t)
    :=
    X_j(t) X_k(t)  = \alpha_{jk}\cI + \sum_{\ell=1}^n\beta_{jk\ell} X_\ell(t),
    \qquad
    j,k=1, \ldots, n,
\end{equation}
which holds at any time $t\> 0$. Here, $\cI$ is the identity operator on $\fH$, so that the right-hand side of (\ref{XXX}) is an affine function of the system variables whose coefficients (the structure constants) comprise a matrix
$\alpha:= (\alpha_{jk})_{1\< j,k\< n} \in \mC^{n\x n}$
  and an  array  $\beta:= (\beta_{jk\ell})_{1\< j,k,\ell\< n} \in \mC^{n\x n\x n}$ with ``sections''
\begin{equation}
\label{betell}
    \beta_\ell
    :=
    \beta_{\bullet \bullet \ell}
    :=
    (\beta_{jk\ell})_{1\< j,k\< n}
    \in
    \mC^{n\x n},
    \qquad
      \ell = 1, \ldots, n.
\end{equation}
The matrices $\alpha, \beta_1, \ldots, \beta_n$ are Hermitian (we denote the subspace of Hermitian matrices of order $n$ by $\mH_n$):
\begin{equation}
\label{herm}
  \alpha^* = \alpha,
  \qquad
  \beta_\ell^* = \beta_\ell,
  \qquad
  \ell = 1, \ldots, n
\end{equation}
(with $(\cdot)^*:= {\overline{(\cdot)}}^\rT$  the complex conjugate transpose), so that their real parts
$\Re \alpha$, $\Re \beta_1, \ldots, \Re \beta_n$ are symmetric, and the imaginary parts $\Im \alpha$, $\Im \beta_1, \ldots, \Im \beta_n$ are  antisymmetric. The relations (\ref{XXX}) can be represented
in terms of a vector\footnote{vectors are organised as columns unless indicated otherwise, and the transpose $(\cdot)^{\rT}$ applies to vectors and matrices of operators as if they consisted of scalars} $X:= (X_k)_{1\< k \< n}$ of the system variables (with the time argument omitted for brevity) as
\begin{equation}
\label{XXX1}
    \Xi:=
    (\Xi_{jk})_{1\< j,k\< n}
    =
    XX^\rT =  \alpha\ox \cI + \beta \cdot X,
\end{equation}
where $\ox$ is the tensor product of operators (in particular, the Kronecker product of matrices) or spaces, so that the matrices $\alpha$ and $\alpha\ox \cI$ can be identified with each other, and
\begin{equation}
\label{betaX}
    \beta \cdot X
    :=
    \sum_{\ell = 1}^n
    \beta_\ell  X_\ell
    =
    \begin{bmatrix}
      \beta_{\bullet 1 \bullet} X
      &
      \ldots &
      \beta_{\bullet n \bullet} X
    \end{bmatrix}
\end{equation}
is an $(n\x n)$-matrix of operators. Here, $\beta_\ell X_\ell := (\beta_{jk\ell}X_\ell)_{1\< j,k\< n} = \beta_\ell \ox X_\ell$, and the columns
$
    \beta_{\bullet k \bullet} X = \big(\sum_{\ell = 1}^n\beta_{jk\ell} X_\ell\big)_{1\< j\< n}
$
involve the appropriate sections     of the array $\beta$:
$$    \beta_{\bullet k \bullet}
    :=
    (\beta_{jk\ell})_{1\< j,\ell \< n} \in \mC^{n\x n},
    \qquad
    k = 1, \ldots, n,
$$
which should not be confused with the sections $\beta_1, \ldots, \beta_n$ from  (\ref{betell}). The latter also give rise to a different product of $\beta$ (and similarly, for any other element of $\mC^{n\x n\x n}$) with a vector $u \in \mC^n$:
\begin{equation}
\label{diam}
    \beta \diam u
    :=
    \begin{bmatrix}
        \beta_1 u & \ldots & \beta_n u
    \end{bmatrix}
    \in \mC^{n\x n}.
\end{equation}
As mentioned in \cite{VP_2022}, for any   $u:=(u_k)_{1\< k\< n},v:=(v_k)_{1\< k\< n} \in \mC^n$,  the products (\ref{betaX}), (\ref{diam}) are related by
\begin{equation}
\label{cdotdiam}
  (\beta \cdot u)v
  =
  (\beta\diam v)u
  =
  \begin{bmatrix}
    \beta_1 & \ldots &  \beta_n
  \end{bmatrix}
  (u\ox v),
\end{equation}
which 
extends to any vectors $u$, $v$ of $n$ quantum variables with zero cross-commutator matrix $[u,v^\rT]:= ([u_j,v_k])_{1\< j,k\< n} = 0$, where $[a,b]:= ab-ba$ is the commutator of linear operators.
In view of the self-adjointness of the system variables $X_1, \ldots, X_n$,
the entries $\Xi_{jk}$ of the matrix $\Xi$ in (\ref{XXX1}), defined as operators  on $\fH$  by the first equality in (\ref{XXX}),   satisfy
\begin{equation}
\label{Xi+}
  \Xi_{jk}^\dagger
  =
  \Xi_{kj},
  \qquad
  j,k = 1, \ldots, n,
\end{equation}
where $(\cdot)^\dagger$ is the operator adjoint. With the latter being extended to matrices of operators as the transpose $(\cdot)^\dagger: = ((\cdot)^\#)^\rT$ of the entrywise operator adjoint $(\cdot)^\#$,  the identities (\ref{Xi+}) can be represented as
\begin{equation}
\label{XiXi}
    \Xi^\dagger = \Xi.
\end{equation}
Accordingly, the conditions (\ref{herm})  (together with the self-adjointness of $X_1, \ldots, X_n$)  secure the consistency
\begin{equation}
\label{XX+}
    \Xi^\#
    =
    \overline{\alpha} + \sum_{\ell = 1}^n \overline{\beta_\ell} X_\ell
    =
    \alpha^\rT  + \sum_{\ell = 1}^n\beta_\ell^\rT  X_\ell = \Xi^\rT
\end{equation}
of the relations (\ref{XXX1}), (\ref{XiXi}).
Moreover (see \cite{VP_2022}), the fulfillment of (\ref{herm}) is necessary for  (\ref{XX+}), provided
\begin{equation}
\label{indep}
  {\rm the\ operators}\ \cI, X_1, \ldots, X_n\ {\rm are\ linearly\ independent}
\end{equation}
in the sense of the implication $c_0 + \sum_{\ell =1}^n c_\ell X_\ell = 0 \Longrightarrow c_0 = \ldots = c_n = 0$ for any $c_0, \ldots, c_n \in \mC$. The condition (\ref{indep}) also plays a role for an additional set of constraints on the structure constants,  which makes (\ref{XXX}) consistent with the associativity of the algebra of linear operators on the Hilbert space $\fH$ in application to the system variables:
\begin{equation}
\label{assoc}
    (X_jX_k)X_\ell = X_j(X_kX_\ell),
    \qquad
    j,k,\ell = 1, \ldots, n,
\end{equation}
with the Jacobi identity
$
    [[X_j,X_k],X_\ell] + [[X_k,X_\ell],X_j] + [[X_\ell,X_j],X_k]  = 0
$ being
a corollary from this   associativity \cite{D_2006}.


\begin{theorem}
\label{th:ass}\cite[Theorem 2.1]{VP_2022}
The following equalities are sufficient and, under the condition (\ref{indep}), necessary for the
relations (\ref{XXX}) (or (\ref{XXX1})) to be consistent with the associativity of the operator multiplication:
\begin{align}
\label{con1}
    \sum_{\ell = 1}^n
    (\alpha_{\ell s}
    \beta_{jk\ell}
    -
    \alpha_{j\ell }
    \beta_{ks \ell }) & = 0,\\
\label{con2}
        \alpha_{jk} \delta_{rs}
        -
        \alpha_{ks} \delta_{rj}
        +
        \sum_{\ell = 1}^n
        (\beta_{jk\ell} \beta_{\ell s r}
        -
        \beta_{ks\ell} \beta_{j\ell r}) & = 0,
        \qquad
        j,k,s,r=1, \ldots, n,
\end{align}
where $\delta_{jk}$ is the Kronecker delta. \hfill$\blacksquare$
\end{theorem}

Repeated application of (\ref{XXX}) reduces the degree three monomials on both sides of (\ref{assoc}) to affine functions of the system variables. By Theorem~\ref{th:ass}, these two functions coincide under the constraints (\ref{con1}), (\ref{con2}) on the structure constants in $\alpha$, $\beta$, which are assumed to be satisfied in what follows. A similar unambiguous  reduction to an affine function  holds for any polynomial of the system variables, including monomials of arbitrary degree $r$:
\begin{equation}
\label{red}
    \rprod_{\ell=1}^r
    X_{j_\ell}^{k_\ell}
  =
  \alpha_{j_1,k_1,\ldots, j_r,k_r}
  +
  \beta_{j_1,k_1,\ldots, j_r,k_r}^\rT
  X,
\end{equation}
where $\rprod_{k=1}^r\zeta_k := \zeta_1 \x \ldots \x \zeta_r$ is the rightward-ordered product of linear operators (the order of multiplication is essential in the noncommutative case). The coefficients $\alpha_{j_1,k_1,\ldots, j_r,k_r}\in \mC$ and $ \beta_{j_1,k_1,\ldots, j_r,k_r} \in \mC^n$, with $j_1,\ldots, j_r = 1, \ldots, n$ and positive integers $k_1,\ldots, k_r$,   can be expressed through the structure constants from (\ref{XXX}).
The ambiguity issue does not arise in reducing quadratic functions of the system variables through a unique application of (\ref{XXX1}):
%
$$    X^\rT R X
     =
    \bra R, \alpha\ket_\rF
    +
    \begin{bmatrix}
      \bra R, \beta_1\ket_\rF &
      \ldots &
      \bra R, \beta_n\ket_\rF
    \end{bmatrix}
    X
$$
for any matrix $R:= (r_{jk})_{1\< j,k\< n} = R^\rT \in \mR^{n\x n}$; see \cite[Eq. (2.10)]{VP_2022} and a similar remark in \cite[paragraph~4 on p.~641]{EMPUJ_2016}. Here, use is made of the Frobenius  inner product    $\bra K, N\ket_\rF:= \Tr(K^* N)$ of real or complex matrices,  with $\|K\|_\rF:= \sqrt{\bra K, K\ket_\rF}$ the Frobenius norm.

%
In addition to the relations (\ref{herm}), (\ref{con1}), (\ref{con2}), another property, associated with the algebraic  structure (\ref{XXX}) of the system variables,  is
the positive semi-definiteness of the operator
\begin{equation}
\label{zaz}
    z^* \alpha z + \sum_{\ell=1}^n (z^* \beta_\ell z) X_\ell
    =
    z^* \Xi z
    = \xi_z^\dagger \xi_z \succcurlyeq 0 ,
\end{equation}
with $\xi_z := z^\rT X $, which holds for any     $z \in \mC^n$ and follows from (\ref{XXX1}).  A related corollary of (\ref{XXX1})  is provided by
$$
    \Tr \alpha +
    \tau^\rT X
    =
    X^\rT X
    =
    \sum_{k=1}^n
    X_k^2
     \succcurlyeq 0 ,
     \qquad
     \tau:= (\tau_\ell)_{1\< \ell \< n},
$$
where $\tau \in \mR^n$ is a vector formed from the traces of the Hermitian sections (\ref{betell}) of the array $\beta$:
\begin{equation}
\label{tau}
      \tau
      :=
      (\tau_k)_{1\< k \< n},
      \qquad
      \tau_k
      :=
      \Tr \beta_k.
\end{equation}
Positive semi-definiteness arguments also play a role in establishing the following bounds on the norms of the system variables $X_1, \ldots, X_n$ as operators on the Hilbert space $\fH$.
%
%
%
%
%

\begin{theorem}
\label{th:bound}\cite[Theorem 2.2]{VP_2022}
Under the conditions (\ref{herm}), (\ref{XXX1}), the induced norms of the system variables satisfy
\begin{equation*}
\label{Xnorm}
  \|X_k\| \< \frac{1}{2} |\tau_k| + \gamma,
  \qquad
  k = 1, \ldots, n,
\end{equation*}
 where $\tau$ is given by (\ref{tau}), and
\begin{equation*}
\label{rad}
  \gamma := \sqrt{\Tr \alpha + \frac{1}{4}|\tau|^2}.
\end{equation*}
\hfill$\blacksquare$
\end{theorem}

Another corollary of the algebraic structure (\ref{XXX}) is concerned with a particular yet important class of quantum variables.

\begin{theorem}
\label{th:alfpos}
If,  in addition to (\ref{herm}), (\ref{XXX1}), the Hilbert space $\fH$ is finite dimensional  and  the system variables are traceless,
\begin{equation}
\label{TrX}
  \Tr X_k = 0,
  \qquad
  k = 1, \ldots, n,
\end{equation}
or $\fH$ is infinite-dimensional and the system variables are trace-class operators,
then the matrix $\alpha$ is positive semi-definite. \hfill$\square$
\end{theorem}
\begin{proof}
Let $\psi_1, \ldots, \psi_N$ be an arbitrary orthonormal basis in the Hilbert   space $\fH$,  where $N:= \dim\fH$ is its dimension which can be infinite (in which case, the basis forms an infinite sequence $(\psi_k)_{k\> 1}$ in $\fH$). Then it follows from (\ref{zaz}) that
\begin{align}
\nonumber
    0 & \<
    \sum_{k=1}^r
    \Big\langle
        \psi_k
        \, \Big|\,
        z^* \alpha z  + \sum_{\ell=1}^n (z^* \beta_\ell z) X_\ell
        \, \Big|\,
        \psi_k
    \Big\rangle\\
\label{zaz1}
    & =
    z^* \alpha z r
    +
    \sum_{\ell=1}^n
    z^* \beta_\ell z
    \sum_{k=1}^r
    \bra
        \psi_k
        \mid
        X_\ell
        \mid
        \psi_k
    \ket,
    \qquad
    z \in \mC^n ,
\end{align}
for any positive integer  $r$, satisfying $r\< N$ if $N< +\infty$. Here, use is made of the quantum mechanical bra-ket notation \cite{S_1994}, including the inner product    $\bra \cdot \mid \cdot\ket$ in the Hilbert space $\fH$. Now, if $\fH$ is finite-dimensional, that is, $N <+\infty$, then by letting $r=N$ and using the relation $    \sum_{k=1}^N
    \bra
        \psi_k
        \mid
        X_\ell
        \mid
        \psi_k
    \ket
    =
    \Tr X_\ell
$, it follows from (\ref{zaz1}), (\ref{TrX}) that
$$
    0 \<
    z^* \alpha z N
    +
    \sum_{\ell=1}^n
    z^* \beta_\ell z
    \Tr X_\ell =     z^* \alpha z N,
$$
and hence, $z^* \alpha z \> 0$ for any $z \in \mC^n$, thus establishing $\alpha \succcurlyeq 0$. By a similar reasoning, in the infinite-dimensional case, when  $N=+\infty$, a combination of (\ref{zaz1}) with the assumption on $X_1, \ldots, X_n$ being trace-class operators (not necessarily satisfying (\ref{TrX})) leads to
$$
    z^* \alpha z
    \>
    -\frac{1}{r}
    \sum_{\ell=1}^n
    z^* \beta_\ell z
    \sum_{k=1}^r
    \bra
        \psi_k
        \mid
        X_\ell
        \mid
        \psi_k
    \ket
    \to 0,
    \qquad
    {\rm as}\
    r\to +\infty,
$$
due to $    \lim_{r\to +\infty}\sum_{k=1}^r
    \bra
        \psi_k
        \mid
        X_\ell
        \mid
        \psi_k
    \ket
    =
    \Tr X_\ell
$, and hence, $z^* \alpha z \> 0$ for all $z \in \mC^n$, whereby $\alpha \succcurlyeq 0$.
\end{proof}

The algebraic structure (\ref{XXX}) (or  (\ref{XXX1})) implies the following one-point CCRs for the system variables:
\begin{align}
\nonumber
    [X,X^\rT]
     & =
    ([X_j,X_k])_{1\< j,k\< n}
    =
    (\Xi_{jk}-\Xi_{kj})_{1\< j,k\< n}    \\
\nonumber
    & =
    \Xi - \Xi^\rT
    =
    \Xi- \Xi^\#\\
\label{XCCR1}
    & =
    2i \Im \Xi
    =
    2i (\Im \alpha + (\Im\beta)\cdot X),
\end{align}
where use is also made of the identity $\Xi^\rT = (\Xi^\dagger)^\rT = \Xi^\#$ in view of (\ref{XiXi}) along with the imaginary part $\Im(\cdot)$ extended from complex numbers to matrices of quantum variables as $\Im \zeta : = \frac{1}{2i}(\zeta - \zeta^\#)$.
Similarly to \cite{VP_2022}, it is assumed for what follows that
\begin{equation}
\label{Imalpha0}
    \Im \alpha = 0
\end{equation}
(so that the Hermitian matrix $\alpha$ is real symmetric),
in which case, the CCRs (\ref{XCCR1}) take the form
\begin{equation}
\label{XCCRTheta}
    [X,X^\rT]
    =
    2i \Theta \cdot X
    =
    2i
    \sum_{\ell=1}^n
    \Theta_\ell X_\ell,
\end{equation}
where
\begin{equation}
\label{Theta}
    \Theta : =  (\theta_{jk\ell})_{1\< j,k,\ell \< n}:=  \Im \beta
\end{equation}
is a real $(n\x n\x n)$-array whose sections
\begin{equation}
\label{Thetaell}
    \Theta_\ell
    :=
    \theta_{\bullet\bullet\ell}
    =
    (\theta_{jk\ell})_{1\< j,k\< n}
    =
    \Im \beta_\ell
    \in
    \mR^{n\x n}
\end{equation}
are antisymmetric for all $\ell = 1, \ldots, n$ in view of the Hermitian property of the matrices $\beta_\ell$ (see (\ref{herm})). Under the condition (\ref{Imalpha0}), the matrix $\alpha$ from (\ref{XXX1}) and the CCR array $\Theta$ are related by
\begin{equation}
\label{con1Im0}
    \sum_{\ell = 1}^n
    (\alpha_{\ell s}
    \theta_{jk\ell}
    -
    \alpha_{j\ell }
    \theta_{ks \ell })  = 0,
    \qquad
    j,k,s=1, \ldots, n,
\end{equation}
which is the imaginary part of the equality (\ref{con1}) from Theorem~\ref{th:ass}.
For example,
the algebraic structure (\ref{XXX}) and the conditions (\ref{indep}),  (\ref{Imalpha0}) are satisfied with $n=3$ for the Pauli matrices \cite{S_1994}
\begin{equation}
\label{X123}
    \sigma_1:=
    \begin{bmatrix}
      0 & 1\\
      1 & 0
    \end{bmatrix},
    \qquad
    \sigma_2:=
    \begin{bmatrix}
      0 & -i\\
      i & 0
    \end{bmatrix}
    =
    -i\bJ,
    \qquad
    \sigma_3:=
    \begin{bmatrix}
      1 & 0\\
      0 & -1
    \end{bmatrix}
\end{equation}
(their tracelessness exemplifies applicability of Theorem~\ref{th:alfpos}), where the matrix
\begin{equation}
\label{bJ}
        \bJ
        : =
        {\begin{bmatrix}
        0 & 1\\
        -1 & 0
    \end{bmatrix}}
\end{equation}
spans the one-dimensional subspace of antisymmetric matrices of order 2.
In this case, the structure constants form the identity matrix $\alpha$ of order 3
and an  imaginary array $\beta$:
\begin{equation}
\label{alfbet}
    \alpha = I_3,
    \qquad
    \beta = i \Theta.
\end{equation}
The corresponding array $\Theta \in\{0, \pm1\}^{3\x 3\x 3}$ in (\ref{Theta}) consists of the Levi-Civita symbols  $\theta_{jk\ell} = \eps_{jk\ell}$, with its sections
\begin{equation}
\label{T123}
    \Theta_1 =
    \begin{bmatrix}
     0&      0  &   0\\
     0&     0   &  1\\
     0&    -1   &  0
    \end{bmatrix},
    \qquad
    \Theta_2 =
    \begin{bmatrix}
     0&      0  &   -1\\
     0&     0   &  0\\
     1&    0   &  0
    \end{bmatrix},
    \qquad
    \Theta_3 =
    \begin{bmatrix}
     0&      1  &   0\\
     -1&     0   &  0\\
     0&    0   &  0
    \end{bmatrix}
\end{equation}
in (\ref{Thetaell})
forming a basis in the subspace of antisymmetric $(3\x 3)$-matrices. Accordingly,  $(\Theta\diam u)v$ is the cross product of vectors $u,v\in \mR^3$,  where
\begin{equation}
\label{Thetadiam}
    \Theta \diam u
    =
    \begin{bmatrix}
      \Theta_1 u &
      \Theta_2 u &
    \Theta_3 u
    \end{bmatrix}
    =
    \begin{bmatrix}
      0 & -u_3 & u_2\\
      u_3 & 0 & -u_1\\
      -u_2 & u_1 & 0
    \end{bmatrix},
    \qquad
    u:= (u_k)_{1\< k\< 3} \in \mR^3,
\end{equation}
describes the  infinitesimal generator of rotations in $\mR^3$ with the angular velocity vector $u$
in view of (\ref{diam}), (\ref{T123}).
The quadruple $\{I_2, \sigma_1, \sigma_2, \sigma_3\}$ of the identity matrix $I_2$ of order 2  and  the Pauli matrices (\ref{X123}) is a basis in the four-dimensional real space $\mH_2$ of Hermitian $(2\x 2)$-matrices which describe self-adjoint operators on the Hilbert space $\mC^2$. The latter  is the simplest of the state spaces for finite-level quantum systems such as the electron spin (an intrinsic quantum angular momentum) interacting with an electromagnetic field.

\section{Quasilinear quantum stochastic system}
\label{sec:QSS}

The evolution of the system variables $X_1, \ldots, X_n$  described in Section~\ref{sec:var}  is driven by the internal dynamics of the quantum system being considered and its interaction with an external bosonic field.  In more complicated settings of open quantum dynamics \cite{BP_2006}, the environment can also include other quantum or classical systems.
Unlike the finite-level variables, the quantum mechanical counterparts of classical fields are organised as unbounded operators  on infinite-dimensional Hilbert spaces.  The quantum stochastic calculus \cite{HP_1984,P_1992} models a multichannel input field by a quantum Wiener process $W: = (W_k)_{1\< k\< m}$ consisting of
an even number $m$ of components $W_1(t), \ldots, W_m(t)$ which are time-varying self-adjoint operators on a symmetric Fock space $\fF$. The latter is equipped with an increasing family $(\fF_t)_{t\> 0}$   of subspaces playing the role of a filtration for $\fF$. In contrast to the standard Wiener process \cite{KS_1991} in $\mR^m$ with the identity diffusion matrix $I_m$, its quantum analogue $W$ has different Ito relations for the future-pointing increments:
\begin{equation}
\label{dWdW_Omega_J_bJ}
    \rd W\rd W^{\rT}
    :=
    \Omega \rd t,
    \qquad
    \Omega
    :=
    I_m + iJ,
    \qquad
        J
        :=
        \bJ\ox I_{m/2},
\end{equation}
where $\bJ$  is given by (\ref{bJ}).
Here, the quantum Ito matrix $\Omega$ is  a complex positive semi-definite Hermitian matrix  of order $m$ (we denote the set of such matrices by $\mH_m^+$) with a nonzero imaginary part $J$  (which is  an orthogonal antisymmetric matrix: $J^2 =-I_m$). The property $\Im \Omega = J\ne 0$ is related to the noncommutative nature of the quantum Wiener process   $W$ in view of the two-point CCRs
\begin{equation*}
\label{WWst}
    [W(s), W(t)^{\rT}]
     = 2i\min(s,t)J ,
    \qquad
    s,t\>0,
\end{equation*}
which are similar to those for the quantum mechanical positions and momenta \cite{GZ_2004}  (or the related annihilation and creation operators)  on the Schwartz space. While a Hilbert space $\fH_0$, which accommodates the initial system variables $X_1(0), \ldots, X_n(0)$,  suffices for  the system dynamics in isolation from the environment, the system-field interaction gives rise to the tensor-product space $\fH:= \fH_0 \ox \fF$ for the system variables, endowed with the filtration
\begin{equation}
\label{fHt}
    \fH_t:= \fH_0\ox \fF_t,
    \qquad
    t \> 0.
\end{equation}
The vector $X$
of the system variables evolves in time according to the  Heisenberg picture of quantum dynamics, and this evolution
is governed by a Hudson-Parthasarathy QSDE \cite{HP_1984,P_1992}
\begin{equation}
\label{dX}
    \rd X
    =
    \cG(X)\rd t  - i[X,L^{\rT}]\rd W.
\end{equation}
Here, the
drift vector
$    \cG(X)
$
 results from the entrywise application of
the Gorini-Kossakowski-Sudar\-shan-Lindblad generator    \cite{GKS_1976,L_1976}, which  acts
on a system operator $\xi$ (a function of the system variables $X_1, \ldots, X_n$) as
\begin{equation}
\label{cG}
\cG(\xi)
   := i[H,\xi]
     +
     \cD(\xi)
\end{equation}
and provides a quantum counterpart of the infinitesimal generators of classical Markov processes \cite{KS_1991}, where $\cD$ is the decoherence superoperator given by
\begin{equation}
\label{cD}
    \cD(\xi)
    :=
     \frac{1}{2}
    ([L^\rT,\xi]\Omega L + L^\rT \Omega [\xi,L]).
\end{equation}
The drift $\cG(X)$ and the dispersion $(n\x m)$-matrix $-i [X, L^{\rT}]$ in (\ref{dX}) consist  of self-adjoint operators on $\fH$ and involve the quantum Ito matrix $\Omega$ from (\ref{dWdW_Omega_J_bJ}) along with the system Hamiltonian $H$ and the vector $L:= (L_k)_{1\< k \< m}$ of system-field coupling operators. Both $H$   and $L_1, \ldots, L_m$  are self-adjoint operators on the system-field space  $\fH$,  which capture the energetics of the system and its interaction with the external field (in particular, $H$ represents the internal energy of the system).   These energy operators are functions (for example, polynomials) of the system variables $X_1, \ldots, X_n$.
Regardless of a particular form of the energy operators, the stochastic flow,  generated by the QSDE (\ref{dX}),  preserves the algebraic structure (\ref{XXX}) of the system variables \cite[Eq.~(3.12)]{VP_2022}.
Since, as mentioned in Section~\ref{sec:var}, due to this algebraic structure, polynomial (and  more general) functions of the system variables reduce to affine functions,  we will be concerned, without loss of generality,  with the case of \cite[Theorem 6.1]{EMPUJ_2016} when the Hamiltonian and the coupling operators
are affine functions of the system variables:
\begin{equation}
\label{H_LM}
    H
    =
    E^\rT X,
    \qquad
    L
    =
    MX + N,
\end{equation}
where $E\in \mR^n$, $M \in \mR^{m\x n}$, $N \in \mR^m$ are the energy and coupling parameters.  Adding a term $c\cI$, with an arbitrary constant $c\in \mR$,   to the Hamiltonian $H$ is irrelevant because $H$ enters the generator (\ref{cG}) only through the commutator. Substitution of (\ref{H_LM}) into (\ref{cD}), (\ref{cG}), (\ref{dX}), combined with the algebraic structure (\ref{XXX}) of the system variables,  leads to a quasilinear QSDE whose parameters are computed in the following theorem,  which is similar to \cite[Lemma~4.2 and Theorem~6.1]{EMPUJ_2016} and is formulated for completeness.

\begin{theorem}
\label{th:QSDE} \cite[Theorem 3.1]{VP_2022}
The QSDE (\ref{dX})
for the open quantum system with the Hamiltonian and coupling operators (\ref{H_LM}) and the dynamic variables satisfying (\ref{XXX}) along with (\ref{herm}), (\ref{Imalpha0}), takes the form
\begin{equation}
\label{dX1}
  \rd X   = (AX + b) \rd t + B(X)\rd W. 
\end{equation}
Here, $A \in \mR^{n\x n}$, $b \in \mR^n$ are a matrix and a vector of coefficients, and $B(X)$ is an $(n\x m)$-matrix  of self-adjoint operators, which depend linearly on the system variables:
\begin{align}
\label{A}
    A
    & :=
        2
        \Theta \diam (E + M^\rT JN)
    +
    2
    \sum_{\ell = 1}^n
    \Theta_\ell
    M^\rT
    (
        M\theta_{\ell\bullet \bullet}
        +
        J M\Re \beta_{\ell\bullet \bullet}
    ),\\
\label{b}
    b
    & :=
    2
    \sum_{\ell = 1}^n
    \Theta_\ell
    M^\rT
    JM\alpha_{\bullet \ell},\\
\label{BX}
    B(X)
    & := 2(\Theta \cdot X)M^\rT,
\end{align}
where the array $\Theta$ is given by (\ref{Theta}), and use is made of its product (\ref{betaX}) with $X$ from (\ref{XCCRTheta}) and the product (\ref{diam}) with the vector $E + M^\rT J N \in \mR^n$.
 \hfill$\blacksquare$
\end{theorem}

In the absence of system-field coupling, when $M=0$, $N=0$ (and so $L=0$ in (\ref{H_LM})),  the matrix $A$ in  (\ref{A}) reduces to
\begin{equation}
\label{A0}
  A_0 := 2
        \Theta \diam E,
\end{equation}
while both $b$ and $B(X)$ in (\ref{b}), (\ref{BX}) vanish. Accordingly, with the decoherence superoperator (\ref{cD}) becoming zero and the generator (\ref{cG}) reducing to $\cG = i[H,\cdot]$ in this case,
the QSDEs (\ref{dX}), (\ref{dX1}) lose the diffusion term and take the form of a linear ODE
\begin{equation}
\label{Xdot}
  \dot{X}
  =
  i[H,X]
  =
  A_0X,
\end{equation}
where $\dot{(\ )}:= \d_t(\ )$ is the time derivative. In this isolated dynamics case,
the Hamiltonian $H$ is a conserved operator regardless of its particular dependence on $X$.  This can also be obtained directly from  (\ref{H_LM}), (\ref{Xdot})  as
\begin{equation}
\label{Hdot}
    \dot{H}
    =
    E^\rT \dot{X}
    =
    E^\rT A_0 X
     =
     0,
\end{equation}
since the matrix (\ref{A0}) satisfies
\begin{equation}
\label{EA0}
    E^\rT A_0
    =
     2
    \begin{bmatrix}
        E^\rT\Theta_1 E & \ldots & E^\rT\Theta_n E
    \end{bmatrix}
     =
     0
\end{equation}
in view of (\ref{diam}) and the relations $E^\rT \Theta_\ell E = 0$ for any $E \in \mR^n$, which  follow from the antisymmetry of the sections $\Theta_1, \ldots, \Theta_n$ of the array $\Theta$ in (\ref{Theta}). The property $\im A_0:= A_0 \mR^n \subset E^{\bot}$ (that $A_0$ maps $\mR^n$ into the hyperplane orthogonal to the vector $E$) also follows from (\ref{cdotdiam}) as
$$
    A_0 u = 2(\Theta \diam E) u = 2(\Theta \cdot u) E \subset E^\bot,
    \qquad
    u \in \mR^n,
$$
since the matrix $\Theta \cdot u$ inherits antisymmetry from (\ref{Thetaell}).
In the presence of system-field coupling, the Hamiltonian $H$ is no longer preserved and, in contrast to (\ref{Hdot}),  evolves according to the QSDE
\begin{align}
\nonumber
    \rd H
    &=
    E^\rT \rd X\\
\nonumber
    & =
    E^\rT ((AX + b) \rd t + B(X)\rd W)\\
\label{dH}
    & =
    E^\rT (\wt{A}X + b) \rd t + E^\rT B(X)\rd W,
\end{align}
obtained by using (\ref{dX1}) and the relation (\ref{EA0}) 
along with the remaining part of the matrix (\ref{A}):
\begin{align}
\nonumber
    \wt{A}
    & :=
    A-A_0\\
\label{At}
    & =
        2
        \Theta \diam (M^\rT JN)
    +
    2
    \sum_{\ell = 1}^n
    \Theta_\ell
    M^\rT
    (
        M\theta_{\ell\bullet \bullet}
        +
        J M\Re \beta_{\ell\bullet \bullet}),
\end{align}
which is a quadratic function of the coupling parameters $M$, $N$. The influence of the external field on the internal energy of the system can also be considered in terms of averaged characteristics using the quantum expectation \cite{H_2001}
\begin{equation*}
\label{bE}
    \bE \zeta := \Tr(\rho\zeta)
\end{equation*}
over a density operator $\rho = \rho^\dagger \succcurlyeq 0$ of unit trace $\Tr \rho =1$  on the space $\fH$, which specifies the system-field quantum state and hence, the statistical properties of quantum variables on $\fH$. In what follows, the system-field density operator is assumed to be the tensor product
\begin{equation}
\label{rho}
    \rho := \rho_0\ox \ups
\end{equation}
of the initial system state $\rho_0$ on $\fH_0$ and the vacuum  state $\ups$ for the input field on the Fock space $\fF$. The latter is specified by the quasi-characteristic functional  \cite{CH_1971,HP_1984,P_1992} of the quantum Wiener process $W$ as
\begin{equation}
\label{vac}
    \bE \re^{i\int_0^t u(s)^\rT\rd W(s) } = \re^{-\frac{1}{2}\int_0^t|u(s)|^2\rd s},
    \qquad
    t\> 0,
\end{equation}
for locally square integrable functions $u: \mR_+\to \mR^m$. In the case of vacuum input fields, the diffusion term $B(X)\rd W$ is a martingale part of the QSDE (\ref{dX1}) and does not contribute to the averaging of (\ref{dH}), which yields
\begin{equation}
\label{EHdot}
    (\bE H)^{^\centerdot}
    =
    E^\rT \dot{\mu}
    =
    E^\rT (\wt{A}\mu + b).
\end{equation}
Here,
\begin{equation}
\label{mu}
      \mu
      :=
      \bE X \in \mR^n
\end{equation}
is the mean vector for the system variables,  whose evolution is governed by the ODE
\begin{equation}
\label{EXdot}
    \dot{\mu }= A\mu + b,
\end{equation}
obtained by averaging the QSDE (\ref{dX1}). The quadratic dependence of $\wt{A}$  in (\ref{At}) and $b$ in (\ref{b}) on the coupling parameters $M$, $N$ is inherited by the quantity $(\bE H)^{^\centerdot}$ in (\ref{EHdot}),  which describes the averaged rate of work of the external field on the system.
 In the steady-state regime, when the matrix $A$ in (\ref{A}) is Hurwitz and the mean vector $\mu$ in (\ref{mu}) is the unique equilibrium solution
 \begin{equation}
\label{mu*}
  \mu_* := \lim_{t\to +\infty}\mu(t) = -A^{-1} b
\end{equation}
of the ODE (\ref{EXdot}), the relation (\ref{EHdot}) implies that $(\bE H)^{^\centerdot} = 0$, and hence, the Hamiltonian $H$, despite not being preserved, has a constant mean value $\bE H$. The steady-state mean vector $\mu_*$ in (\ref{mu*}) specifies the equilibrium mixed moments of the system variables of higher orders:
\begin{equation}
\label{mured}
    \lim_{t\to +\infty}
    \bE
    \rprod_{\ell=1}^r
    X_{j_\ell}(t)^{k_\ell}
  =
  \alpha_{j_1,k_1,\ldots, j_r,k_r}
  +
  \beta_{j_1,k_1,\ldots, j_r,k_r}^\rT
  \mu_*
\end{equation}
in view of the reduction (\ref{red}) (with the latter being a corollary  of the algebraic structure (\ref{XXX}) regardless of the system dynamics).  The statistical properties of the system variables in the invariant quantum state, including (\ref{mured}), are captured in the quasi-characteristic function (QCF):
\begin{equation}
\label{QCF*}
    \lim_{t\to +\infty}
    \bE \re^{iu^\rT X(t)}
    =
        \begin{bmatrix}
        1 & \bzero_n^\rT
    \end{bmatrix}
    \exp\left(
    i
    \begin{bmatrix}
      0 & u^\rT \\
      \alpha u &  \beta \diam u
    \end{bmatrix}
    \right)
    \begin{bmatrix}
      1\\
      \mu_*
    \end{bmatrix},
    \qquad
    u \in \mR^n,
\end{equation}
computed in  \cite[Eq. (4.10)]{VP_2022} using \cite[Theorem 4.1]{VP_2022}, where $\bzero_n$ is the column-vector of $n$ zeros. In view of the ODE (\ref{EXdot}), the convergence of the moments to their equilibrium values in  (\ref{mu*})--(\ref{QCF*}) holds at an exponential rate depending on the quantity
\begin{equation}
\label{stab}
    \sigma(A)
    :=
    \ln \br(\re^A) = \max_{1\< k \< n}\Re \lambda_k < 0,
\end{equation}
where $\br(\cdot)$ is the spectral radius, and $\lambda_1, \ldots, \lambda_n$ are the eigenvalues of the matrix  $A$ (which is assumed to be Hurwitz here). The more negative $\sigma(A)$ is, the faster the convergence. This exponentially fast convergence can be used for the generation of invariant quantum states with given moments through a sufficiently long run of the system, similar to \cite{Y_2012,Y_2009}. This  state generation procedure exploits the dissipative nature of the system-field interaction. On the other hand, in comparison with the isolated dynamics, such interaction is accompanied by a ``loss of quantumness'' over time, which is  interpreted as decoherence.

While the relations (\ref{mu*})--(\ref{QCF*}) are concerned with one-point moments of the system variables,
the quasilinearity of the QSDE (\ref{dX1}) also leads to tractable dynamics of multipoint  mixed moments of the system variables (at different times) as discussed in  \cite[Section 4]{VP_2022}. 
This tractability is similar to yet different from the case of  linear QSDEs for OQHOs with Gaussian states \cite{NY_2017,KRP_2010,P_2017} and is closely related to the existence of fundamental solutions  of (\ref{dX1}) in the form of time-ordered operator exponentials  \cite{H_1996}, which is described below and holds regardless of a particular quantum state.


\begin{theorem}
\label{th:var} \cite[Theorem 3.2]{VP_2022}
Under the conditions of Theorem~\ref{th:QSDE}, the system variables governed by (\ref{dX1})   satisfy
\begin{equation}
\label{Xts}
    X(t) = \cE(t,s)X(s) + \int_s^t \cE(t,\tau)\rd \tau b,
    \qquad
    t \> s \> 0,
\end{equation}
where
\begin{equation}
\label{Ets}
  \cE(t,s)
  :=
  (\cE_{jk}(t,s))_{1\< j,k\< n}
  :=
    \lexp
  \int_s^t
  (A\rd \tau + 2\Theta \diam (M^\rT \rd W(\tau)))
\end{equation}
is a leftwards time-ordered operator exponential, and use is made of (\ref{diam}), (\ref{Theta}),  (\ref{A}), (\ref{b}). \hfill$\blacksquare$
\end{theorem}

In contrast to the exponentials of real or complex  matrices (describing the fundamental solutions of linear ODEs with constant coefficients), $\cE(t,s)$ in (\ref{Ets})  is an $(n\x n)$-matrix of self-adjoint operators on the orthogonal complement $\cF_s^\bot \bigcap \cF_t$ of the Fock subspace $\cF_s$ to $\cF_t$ for any $t\> s \> 0$. This is a corollary of the continuous tensor-product structure of the Fock space \cite{PS_1972}. In particular, the entries of $\cE(t,s)$ commute with the past system variables in $X(s)$ and future increments $\rd W(t)$ of the quantum Wiener process:
\begin{equation}
\label{EXcomm}
  [\cE_k(t,s), X(s)^\rT] = 0,
  \quad
  [\cE_k(t,s), \rd W(t)^\rT] = 0,
  \quad
  t\> s\> 0,
  \
  k = 1, \ldots, n,
\end{equation}
where $\cE_k(t,s):= (\cE_{jk}(t,s))_{1\< j \< n}$ is the $k$th column of the matrix $\cE(t,s)$, which,  in view of (\ref{Xts}),  pertains to the response of $X(t)$ to $X_k(s)$.
 The exponential $\cE(t,s)$ in (\ref{Ets}) is the fundamental solution for the  homogeneous version of the QSDE (\ref{dX1}),   with $b$ removed,  satisfying the initial value problem
\begin{equation}
\label{dEts}
    \rd_t \cE(t,s)
    =
    A\cE(t,s)\rd t + B(\cE(t,s))\rd W(t),
    \qquad
    t \> s \> 0,
    \qquad
    \cE(s,s) := I_n
\end{equation}
(see the proof of \cite[Theorem 3.6]{VP_2022}).
Here, the map $B$, defined in (\ref{BX}) on $n$-dimensional column-vectors, is extended in $B(\cE(t,s))\rd W$  by acting on the  columns of $\cE(t,s)$  as
\begin{equation}
\label{BEW}
    B(\cE(t,s))\rd W :=
    \begin{bmatrix}
    B(\cE_1(t,s))\rd W
    &
    \ldots&
    B(\cE_n(t,s))\rd W
    \end{bmatrix}.
\end{equation}
An equivalent  column-wise form of (\ref{dEts}) is organised as a set of $n$ identical QSDEs with different initial conditions and driven by the common quantum Wiener process $W$:
\begin{align}
\nonumber
    \rd_t \cE_k(t,s)
    = &
    A\cE_k(t,s)\rd t
    +
    B(\cE_k(t,s))\rd W(t),\\
\label{dEtsk}
    & t \> s \> 0,
    \qquad
    \cE_k(s,s)
    :=
    \delta_k,
    \qquad
    k = 1, \ldots, n,
\end{align}
where $\delta_k:= (\delta_{jk})_{1\< j \< n}$ is the $k$th standard basis vector in $\mR^n$.
Since the right-hand side of the QSDE in (\ref{dEts}) depends linearly on $\cE(t,s)$, then (similarly to classical linear SDEs with state-dependent noise \cite{W_1967}) the first, second and higher-order moments \cite{VP_2022} of its solution $\cE(t,s)$ have tractable dynamics in the case of vacuum fields specified by (\ref{rho}), (\ref{vac}). This tractability includes the possibility to compute relevant Lyapunov exponents. 

\section{Exponential decay in two-point CCRs}\label{sec:decay}

As established in \cite{VP_2022}, the relation (\ref{Xts})  between the system variables at different moments of time, combined with the properties (\ref{EXcomm}) of the exponential $\cE(t,s)$ from (\ref{Ets}),  lead to the two-point CCRs
\begin{equation}
\label{XXcommts}
    [X(t), X(s)^\rT]
%
    =
    2i \cE(t,s) (\Theta \cdot X(s)),
    \qquad
    t \> s\> 0,
\end{equation}
which extend their one-point counterpart (\ref{XCCRTheta}). These  CCRs are a consequence of the commutation structure of the system variables and  the external fields regardless of a particular quantum state or the Hurwitz  property of $A$. However, the latter and the presence of the operator exponential $\cE(t,s)$ suggest an exponentially fast  decay in the two-point CCRs (\ref{XXcommts}) as $t-s\to +\infty$, so that the system variables $X_k(t)$ at a distant instant $t$ become ``asymptotically commuting'' with (and hence, classical with respect to) $X_j(s)$. Such a decay in the two-point CCRs can be regarded as a manifestation of quantum decoherence in the system caused by its dissipative coupling to the environment.

The rate of decay in (\ref{XXcommts}) can be quantified in terms of moments over the system-field state (\ref{rho}), which corresponds to vacuum input fields. Similarly to the proof of \cite[Theorem 4.2]{VP_2022}, the structure of the time-ordered operator exponential $\cE(t,s)$ in (\ref{Ets}) and the continuous tensor-product structure of the Fock space $\fF$ and the vacuum state $\ups$ imply that
\begin{equation}
\label{EEE1}
    \bE (\cE(t,s)\eta) = \bE \cE(t,s) \bE \eta,
    \qquad
    t\> s\> 0,
\end{equation}
for any quantum variable $\eta$ adapted to the system-field subspace $\fH_s$ from (\ref{fHt}). Indeed, any such $\eta$ commutes with and is statistically independent of the entries of $\cE(t,s)$. The averaging of the QSDE (\ref{dEts}) (with the martingale part $B(\cE(t,s))\rd W(t)$ in (\ref{BEW}) not contributing to the average) leads to the ODE $\d_t \bE \cE(t,s) = A \bE \cE(t,s)$ and hence,
\begin{equation}
\label{EEE2}
    \bE \cE(t,s) = \re^{(t-s)A},
    \qquad
    t \> s\> 0,
\end{equation}
in view of the initial condition in (\ref{dEts}). In particular, by applying (\ref{EEE1}) to the system variables $\eta:= X_k(s)$ for all $k =1, \ldots, n$ (in view of (\ref{EXcomm}))  in combination with (\ref{EEE2}), it follows from (\ref{XXcommts}) that
\begin{align}
\nonumber
    \bE [X(t), X(s)^\rT]
%
    & =
    2i \bE \cE(t,s) (\Theta \cdot \bE X(s))\\
\label{EXXcommts}
    & =
    2i \re^{(t-s)A}(\Theta \cdot \mu(s)),
    \qquad
    t \> s\> 0,
\end{align}
where use is also made of the mean vector (\ref{mu}). Therefore, the Hurwitz property of the matrix $A$ implies not only an exponentially fast convergence to the invariant  quantum state  in (\ref{mu*})--(\ref{QCF*}) with the leading Lyapunov exponent (\ref{stab}), but also an exponential decay in the two-point CCRs (\ref{XXcommts}) in the sense of the mean values (\ref{EXXcommts}):
\begin{equation}
\label{Lyapexp0}
    \limsup_{t\to +\infty}
    \Big(
    \frac{1}{t-s}
    \ln \|\bE [X(t), X(s)^\rT]\|_\rF
    \Big)
    \<
    \sigma(A) <0,
    \qquad
    s\> 0,
\end{equation}
where any other matrix norm can also be used instead of the Frobenius norm $\|\cdot \|_\rF$ without  affecting the limit.

However, quantification of decoherence as the two-point CCR decay in terms of the second moments of $\cE(t,s)$ can lead to different rates in comparison with (\ref{stab}). This discrepancy comes from the fact that the algebraic structure (\ref{XXX}) of the system variables does not carry over to the entries of $\cE(t,s)$ and so the second and higher-order moments of $\cE_{jk}(t,s)$ do not reduce to their mean values in (\ref{EEE2}). In order to study the second-order moment dynamics,
we will apply the vectorisation $\vec{(\cdot)}$ (or, interchangeably,  $\vect(\cdot)$) of matrices \cite{M_1988,SIG_1998}  to $\cE(t,s)$ by stacking its columns in a vector
\begin{equation}
\label{Evec}
    \vec{\cE}(t,s)
    :=
    \begin{bmatrix}
      \cE_1(t,s)\\
      \vdots\\
      \cE_n(t,s)
    \end{bmatrix}
\end{equation}
consisting of $n^2$ quantum variables. As a function of $t$,  the quantum process $\vec{\cE}(t,s)$ satisfies
\begin{equation}
\label{dEtsvec}
    \rd_t \vec{\cE}(t,s)
    =
    \cA
    \vec{\cE}(t,s)
    \rd t +
    \cB(\vec{\cE}(t,s)) \rd W(t),
    \qquad
    t \> s \> 0,
    \qquad
    \vec{\cE}(s,s) := \vec{I}_n,
\end{equation}
which is a vectorised representation of the initial value problem (\ref{dEts}) using (\ref{BEW}) along with
\begin{equation*}
\label{cAB1}
    \cA
    :=
    I_n \ox A,
    \qquad
    \cB(\vec{\cE}(t,s))
    :=
    \begin{bmatrix}
      B(\cE_1(t,s))\\
      \vdots\\
      B(\cE_n(t,s))
    \end{bmatrix}.
\end{equation*}
The vectorisation (\ref{Evec}) allows all the second-order  moments for the entries of $\cE(t,s)$ to be captured in
\begin{equation}
\label{Pts}
  P(t,s)
  :=
  \bE (\vec{\cE}(t,s)\vec{\cE}(t,s)^\rT)
  =
(P_{jk}(t,s))_{1\< j,k\< n},
\end{equation}
which is a complex positive semi-definite Hermitian matrix of order $n^2$ with the $(n\x n)$-blocks
\begin{equation}
\label{Pjk}
    P_{jk}(t,s)
    :=
  \bE (\cE_j(t,s)\cE_k(t,s)^\rT)
  =
  P_{kj}(t,s)^*.
\end{equation}
The nonnegative quantity
\begin{equation}
\label{Ptrace}
    \Tr P(t,s)
    =
    \sum_{k= 1}^n
    \Tr P_{kk}(t,s)
    =
    \sum_{j,k= 1}^n
    \bE (\cE_{jk}(t,s)^2)
\end{equation}
can be used for mean square bounds on the CCR decay.  Moreover, since the matrix (\ref{Pts}) satisfies  $P(t,s) = P(t,s)^*\succcurlyeq 0$, then its off-diagonal blocks admit an upper bound in terms of (\ref{Ptrace}):
\begin{align}
\nonumber
    \sum_{1\< j\ne k \< n}
    \|P_{jk}\|_\rF^2
    & =
    2
    \sum_{1\< j< k \< n}
    \|P_{jk}\|_\rF^2\\
\nonumber
    & =
    \|P\|_\rF^2
    -
    \sum_{k=1}^n
    \|P_{kk}\|_\rF^2\\
\nonumber
    & \<
    (\Tr P)^2
    -
    \frac{1}{n}
    \sum_{k=1}^n
    (\Tr P_{kk})^2
    \<
    \Big(
        1 - \frac{1}{n^2}
    \Big)
    (\Tr P)^2    ,
\end{align}
where the time arguments are omitted for brevity, and the inequalities $\frac{1}{r}(\sum_{k=1}^r \gamma_k) ^2\< \sum_{k=1}^r \gamma_k^2 \< (\sum_{k=1}^r \gamma_k) ^2$  for any $\gamma_1, \ldots, \gamma_r \in \mR_+$ and $r\> 1$  are applied to the eigenvalues of the matrix $P$ and its diagonal blocks $P_{kk}$.
The following theorem, which  describes the dynamics of the second moments in (\ref{Pts}), (\ref{Pjk}), employs linear operators
$\Lambda$, $\Phi$,   $U: \mC^{n\x n}\mapsto \mC^{n\x n}$ defined by
\begin{align}
\label{Lambda}
  \Lambda
  & :=
  \Phi + U,\\
\label{Phi}
    \Phi(Z)
    & := A Z + ZA^\rT,\\
\label{V}
    U(Z)
    & :=
    -4
    \sum_{j,k=1}^n
    z_{jk}
    \Theta_j
    M^\rT\Omega  M
    \Theta_k,
    \qquad
    Z:= (z_{jk})_{1\< j,k\< n}\in \mC^{n\x n},
\end{align}
using the sections (\ref{Thetaell}) of the CCR array $\Theta$, the coupling matrix $M$ from (\ref{H_LM}) and the quantum Ito matrix $\Omega$ in (\ref{dWdW_Omega_J_bJ}). These   operators map the subspace of Hermitian matrices of order $n$ into itself: \begin{equation}
\label{invar}
    \Lambda(\mH_n), \Phi(\mH_n), U(\mH_n)\subset \mH_n.
\end{equation}
Furthermore,  the operator $U$ is positive (with respect to the cone $\mH_n^+$) in the sense of the inclusion
\begin{equation}
\label{Vpos}
    U(\mH_n^+) \subset \mH_n^+.
\end{equation}
This property follows from
\begin{align}
\nonumber
    U(zz^*)
    & =
    -4
    \sum_{j,k=1}^n
    z_j\overline{z_k}
    \Theta_j
    M^\rT\Omega  M
    \Theta_k\\
\label{Vzzpos}
    & =
    B(z)\Omega B(z)^*
    \succcurlyeq 0,
    \qquad
    z:= (z_k)_{1\< k\< n}\in \mC^n,
\end{align}
where the map $B$ from (\ref{BX}) is evaluated  at $z$, and the properties $\Omega = \Omega^*\succcurlyeq 0$ of the quantum Ito matrix are used. Indeed, since any matrix $Z \in \mH_n^+$ can be represented as  a linear combination $Z = \sum_{k=1}^n \gamma_k w_kw_k^*$ of rank one matrices $w_kw_k^*$, with $w_k \in \mC^n$ and nonnegative coefficients $\gamma_k \in \mR_+$, then $U(Z) = \sum_{k=1}^n \gamma_k U(w_kw_k^*) \succcurlyeq 0$ due to the linearity of $U$ and (\ref{Vzzpos}). The positivity (\ref{Vpos}) can also be established by representing (\ref{V}) in terms of an auxiliary matrix
\begin{equation}
\label{Delta}
    \Delta
    :=
    \begin{bmatrix}
      \Theta_1 & \ldots & \Theta_n
    \end{bmatrix}
    \in \mR^{n \x n^2}
\end{equation}
as
$$
    U(Z)
    =
    4
    \Delta
    (Z\ox (M^\rT \Omega M))
    \Delta^\rT,
    \qquad
    Z \in \mC^{n\x n},
$$
and using the fact that the Kronecker product of positive semi-definite Hermitian matrices is again such a matrix.

\begin{theorem}
\label{th:Pjkdot}
Suppose the conditions of Theorem~\ref{th:QSDE} are satisfied, and the input field $W$ is in the vacuum state as specified by (\ref{rho}), (\ref{vac}). Then the second moment matrix $P$ in (\ref{Pts}) for the columns of the operator exponential $\cE(t,s)$ in (\ref{Ets}) depends only on the time difference,
\begin{equation}
\label{PPi}
    P(t,s) = \Pi(t-s),
  \qquad
  t \> s \> 0,
\end{equation}
where $\Pi:=(\Pi_{jk})_{1\< j,k\< n}$ is a matrix-valued function of time whose blocks $\Pi_{jk}: \mR_+\to \mC^{n\x n}$ satisfy the initial value problems
\begin{equation}
\label{Pijkdot}
    \dot{\Pi}_{jk}
    =
    \Lambda(\Pi_{jk}),
    \qquad
    \Pi_{jk}(0) = \delta_j\delta_k^\rT,
    \qquad
    j,k=1, \ldots, n,
\end{equation}
where the operator $\Lambda$ is given by (\ref{Lambda}).
\hfill$\square$
\end{theorem}
\begin{proof}
By the QSDEs in (\ref{dEtsk}) (or (\ref{dEtsvec})) and the quantum Ito lemma \cite{HP_1984,P_1992}, the blocks $\cE_j \cE_k^\rT$ of the matrix  $\vec{\cE}\vec{\cE}^\rT$ satisfy
\begin{align}
\nonumber
    \rd_t (\cE_j \cE_k^\rT)
    =&
    (\rd_t \cE_j ) \cE_k ^\rT
    +
    \cE_j
    \rd_t \cE_k ^\rT
    +
    \rd_t \cE_j \rd_t \cE_k ^\rT\\
\nonumber
    = &
    (A \cE_j   \rd t + B(\cE_j  )\rd W)\cE_k  ^\rT\\
\nonumber
    & +
    \cE_j
    (\cE_k  ^\rT A^\rT  \rd t + (\rd W^\rT) B(\cE_k  )^\rT)\\
\nonumber
    & + B(\cE_j  ) \rd W \rd W^\rT B(\cE_k  )^\rT\\
\nonumber
    =& (A \cE_j  \cE_k ^\rT + \cE_j  \cE_k ^\rT A^\rT + B(\cE_j ) \Omega B(\cE_k )^\rT)\rd t\\
\label{dEEjk}
    & +
    B(\cE_j )(\rd W) \cE_k ^\rT
    +
    \cE_j  (\rd W)^\rT B(\cE_k )^\rT,
    \qquad
    j, k= 1, \ldots, n
\end{align}
%
%
%
(the time arguments $t\> s\> 0$ of $\cE(t,s)$, $W(t)$ are omitted for brevity),
where the quantum Ito matrix $\Omega$ from (\ref{dWdW_Omega_J_bJ}) and  the second commutativity from (\ref{EXcomm}) are also used. Since the input field is assumed to be in the vacuum state, the last line of the QSDE (\ref{dEEjk}) describes its martingale part which does not contribute to the quantum average of the right-hand side. Hence, the averaging of (\ref{dEEjk}) shows that the blocks (\ref{Pjk}) satisfy the ODEs
\begin{equation}
\label{Pjkdot1}
    \d_t P_{jk}(t,s)
    =
    \Phi(P_{jk}(t,s)) + \bE (B(\cE_j (t,s))\Omega B(\cE_k (t,s))^{\rT}),
\end{equation}
where use is also made of the operator $\Phi$ from (\ref{Phi}).
Now, for any vectors $\xi:= (\xi_k)_{1\< k \< n}$ and  $\eta:= (\eta_k)_{1\< k \< n}$ of self-adjoint quantum variables with finite second moments,
\begin{align}
\nonumber
    \bE (B(\xi)\Omega B(\eta)^{\rT})
    & =
    4
    \bE((\Theta \cdot \xi)M^\rT \Omega M (\Theta \cdot \eta)^\rT )\\
\nonumber
    & =
    -4
    \bE
    \Big(
    \sum_{j=1}^n \Theta_j\xi_j
    M^\rT \Omega M
    \sum_{k=1}^n \Theta_k\eta_k
    \Big)\\
\label{BOB}
    & =
    -4
    \sum_{j,k=1}^n
    \bE(\xi_j\eta_k)
    \Theta_j
    M^\rT \Omega M
    \Theta_k
    =
    U(\bE(\xi\eta^\rT)),
\end{align}
where (\ref{BX}) is used along with (\ref{betaX}), the antisymmetry of the sections (\ref{Thetaell}) of the array $\Theta$ and the definition (\ref{V}) of the operator $U$.   By applying (\ref{BOB}) to $\xi:= \cE_j (t,s)$, $\eta:= \cE_k (t,s)$, so that $\bE(\xi\eta^\rT) = P_{jk}(t,s)$ in view of (\ref{Pjk}),  and combining the result with (\ref{Pjkdot1}), (\ref{Lambda}), it follows that the function $P_{jk}$ satisfies the initial value problem
\begin{equation}
\label{Pjkdot}
    \d_t P_{jk}(t,s) = \Lambda(P_{jk}(t,s)),
    \qquad
    t \> s\> 0,
    \qquad
    P_{jk}(s,s) = \delta_j \delta_k^\rT,
\end{equation}
for any     $
    j,k=1, \ldots, n$,
where the initial condition is obtained from that in (\ref{dEtsk}). Since the ODE in (\ref{Pjkdot}) is autonomous and the initial condition $\delta_j \delta_k^\rT$ does not depend on $s$, then the solution indeed depends only  on $t-s\> 0$ as $P_{jk}(t,s) = \Pi_{jk}(t-s)$ as described in (\ref{PPi}), (\ref{Pijkdot}).
\end{proof}

The solution of the initial value problem  (\ref{Pijkdot}) can be represented by using an operator exponential on $\mC^{n\x n}$ as
$$
    \Pi_{jk}(t) = \re^{t \Lambda}(\delta_j\delta_k^\rT),
    \qquad
    t \> 0,
    \qquad
    j, k = 1, \ldots, n.
$$
In particular, the time evolution of the diagonal blocks $\Pi_{kk}(t) \in \mH_n^+$ of the matrix $\Pi(t)$ in (\ref{PPi}),  which inherit their Hermitian property and positive semi-definiteness   from the initial conditions $\delta_k\delta_k^\rT$, is completely specified by the restriction $\Lambda|_{\mH_n}$ of the operator $\Lambda$ in (\ref{Lambda})  to the invariant subspace $\mH_n$.
The corresponding Lyapunov exponents satisfy
\begin{equation}
\label{Lyapexp}
    \limsup_{t\to +\infty}
    \Big(
    \frac{1}{t}
    \ln \|\Pi_{kk}(t)\|_\rF
    \Big)
    \<
    \sigma(\Lambda|_{\mH_n}),
    \qquad
    k = 1, \ldots, n,
\end{equation}
where the right-hand side employs the spectral radius of $\Lambda|_{\mH_n}$ defined similarly to (\ref{stab}). Note that, in contrast to the matrix $A$ which can have essentially complex eigenvalues, the spectrum of the operator $\Lambda|_{\mH_n}$ is real. Indeed, if $\lambda \in \mC$ is an eigenvalue of $\Lambda|_{\mH_n}$,  then there exists a  matrix  $Z\in \mH_n \setminus \{0\}$ such that $\Lambda(Z) = \lambda Z$, and hence,  $\lambda Z \in \mH_n$ due to the invariance (\ref{invar}) of the subspace $\mH_n$. Since both $Z$ and $\lambda Z$ are Hermitian, then $2i Z \Im \lambda = \lambda Z -(\lambda Z )^* = 0$,  which,  in view of the condition $Z\ne 0$, implies that $\lambda \in \mR$.

\begin{theorem}
\label{th:Lyapexp}
Suppose the conditions of Theorem~\ref{th:Pjkdot} are satisfied. Then the Lyapunov exponent for the function $\Pi$ in (\ref{PPi}) is related to the spectra of the matrix $A$ in (\ref{A}) and the restriction $\Lambda|_{\mH_n}$ of the  operator $\Lambda$ in (\ref{Lambda}) as
\begin{equation}
\label{PiLyapexp}
    2\sigma(A)
    \<
    \limsup_{t\to +\infty}
    \Big(
    \frac{1}{t}
    \ln
    \Tr \Pi(t)
    \Big)
    \<
    \sigma(\Lambda|_{\mH_n}).
\end{equation}
\hfill$\square$
\end{theorem}
\begin{proof}
The one-parameter group generated by the operator $\Phi$ in (\ref{Phi}) consists of positive operators acting on $\mC^{n\x n}$ as
\begin{equation}
\label{Phiexp}
    \re^{t\Phi}(Z) = \re^{tA} Z \re^{tA^\rT},
    \qquad
    Z \in \mC^{n\x n},
    \quad
    t \in \mR.
\end{equation}
Therefore, the decomposition (\ref{Lambda}) of the operator $\Lambda$ into $\Phi$, $U$ and the variation of constants allow the solution $\Pi_{kk}$ of the initial value problem (\ref{Pijkdot}) to be represented as
\begin{align}
\nonumber
    \Pi_{kk}(t)
    & =
    \re^{t\Lambda}(\Pi_{kk}(0))\\
\nonumber
    & =
    \re^{t\Phi}(\Pi_{kk}(0))
    +
    \int_0^t
    \re^{(t-s)\Phi}(U(\Pi_{kk}(s)))
    \rd s\\
\label{Pikksol}
    & \succcurlyeq
    \re^{t\Phi}(\Pi_{kk}(0))
    =
    \re^{tA}\delta_k\delta_k^\rT \re^{tA^\rT},
    \qquad
    t \> 0,
    \quad
    k = 1, \ldots , n.
\end{align}
Here, the matrix inequality uses
the positivity of the operators $U$ in (\ref{V}) and $\re^{t\Phi}$ in (\ref{Phiexp}) and the property that $\Pi_{kk}$ is $\mH_n^+$-valued. A combination of (\ref{Pikksol}) with the resolution  of the identity  $        \sum_{k=1}^n
        \delta_k\delta_k^\rT = I_n
$ leads to
\begin{equation}
\label{Pitrace}
    \Tr \Pi(t)
    =
    \sum_{k=1}^n
    \Tr \Pi_{kk}(t)
    =
    \Tr \re^{t \Lambda}(I_n)
    \>
    \|\re^{tA}\|_\rF^2,
\end{equation}
and hence,
$$
    \limsup_{t\to +\infty}
    \Big(
    \frac{1}{t}
    \ln
    \Tr \Pi(t)
    \Big)
    \>
    2
    \lim_{t\to +\infty}
    \Big(
    \frac{1}{t}
    \ln
    \|\re^{tA}\|_\rF
    \Big)
    =
    2\sigma(A),
$$
which proves the first inequality in (\ref{PiLyapexp}). The second inequality in (\ref{PiLyapexp}) is similar to (\ref{Lyapexp}) and follows from the second equality in (\ref{Pitrace}) and the invariance of the set $\mH_n^+ \subset \mH_n$ under the operator $\Lambda$.
\end{proof}


By combining (\ref{PiLyapexp}) with (\ref{Ptrace}), (\ref{PPi}), it follows that  the Lyapunov exponent for the root mean square value  of the exponential (\ref{Ets}) satisfies
\begin{align}
\nonumber
    \sigma(A)
    & \<
    \limsup_{t\to +\infty}
    \left(
        \frac{1}{t-s}
        \ln
        \sqrt{    \sum_{j,k= 1}^n
    \bE (\cE_{jk}(t,s)^2)}
    \right)\\
\label{rms}
    & =
    \frac{1}{2}
    \limsup_{t\to +\infty}
    \Big(
    \frac{1}{t}
    \ln
    \Tr \Pi(t)
    \Big)
    \<
    \frac{1}{2}
    \sigma(\Lambda|_{\mH_n}).
\end{align}
Therefore, the condition $\sigma(\Lambda|_{\mH_n})< 0$ (which is stronger than the Hurwitz property of $A$) is sufficient for  an exponential decay in the two-point CCRs in terms of the second moments. The spectrum of the operator $\Lambda$ in (\ref{Lambda}), including that of $\Lambda|_{\mH_n}$,  can be computed by using its vectorised representation
\begin{equation}
\label{vec}
    \vect(\Lambda(Z)) = (A \op A + \Psi) \vec{Z},
    \qquad
    Z \in \mC^{n\x n}.
\end{equation}
Here, $A \op A:= I_n \ox A + A \ox I_n$ is the Kronecker sum of the matrix $A$ with itself, which represents the operator $\Phi$ in (\ref{Phi}), and $\Psi \in \mC^{n^2\x n^2}$ is an auxiliary matrix defined by
\begin{align*}
    \Psi
    & :=
    \begin{bmatrix}
      \Psi_1 & \ldots & \Psi_n
    \end{bmatrix},
    \qquad
    \Psi_k
    :=
    \begin{bmatrix}
      \Psi_{1k} & \ldots & \Psi_{nk}
    \end{bmatrix},\\
    \Psi_{jk}
    & :=
    -4
    \vect(\Theta_j
    M^\rT\Omega  M
    \Theta_k)\in \mC^{n^2} ,
    \qquad
    j,k = 1, \ldots, n,
\end{align*}
in accordance with the structure of the operator $U$ in (\ref{V}), so that $\vect(U(Z)) = \sum_{j,k=1}^n z_{jk}\Psi_{jk}$. Since the vectorization is a linear bijection between  the spaces $\mC^{n\x n}$ and $\mC^{n^2}$, the spectrum of the operator $\Lambda$ coincides with that of the matrix $A \op A + \Psi$ in (\ref{vec}).


\section{A comparison with the oscillatory modes of isolated system dynamics}
\label{sec:iso}

When $A$ is Hurwitz,  the exponentially fast decay (\ref{Lyapexp0}) in the two-point CCRs    (and the related mean square bounds in (\ref{rms}))  is closely related to $X(s)$ having a fading effect on $X(t)$, so that the forced response term in (\ref{Xts}) becomes dominant as $t-s\to +\infty$. This is qualitatively different  from the isolated system dynamics in the absence of coupling, when $M=0$, $N=0$ in (\ref{H_LM}) and the matrix $A$ reduces to $A_0$ in (\ref{A0}). The following theorem is concerned with the oscillatory behaviour of the isolated system.

\begin{theorem}
\label{th:osc}
Suppose the conditions of Theorem~\ref{th:QSDE} are satisfied,  and  the matrix $\alpha$ in (\ref{XXX1}) is positive definite:
 \begin{equation}
\label{alfpos}
    \alpha\succ 0.
\end{equation}
Then for any energy vector $E\in \mR^n$, the spectrum of the matrix $A_0$ in (\ref{A0}) is contained by the imaginary axis $i\mR$ (that is, each eigenvalue is either  purely imaginary or zero). \hfill$\square$
\end{theorem}
\begin{proof}
By using (\ref{A0}) along with (\ref{diam}), (\ref{Imalpha0}),  (\ref{Thetaell}), it follows that the entries of the matrix $A_0\alpha$ are computed as
\begin{equation}
\label{Aajs}
    (A_0 \alpha)_{js}
    =
    \sum_{k,\ell=1}^n
    \theta_{jk\ell}
    E_k
    \alpha_{\ell s},
    \qquad
    j,s = 1, \ldots, n,
\end{equation}
where $E_1, \ldots, E_n$ are the entries of the energy vector $E$. Since the matrix $\alpha$ is symmetric in view of (\ref{herm}), (\ref{Imalpha0}),   then (\ref{Aajs}) leads to
\begin{align}
\nonumber
    (\alpha A_0^\rT)_{js}
    & =
    ((A_0 \alpha)^{\rT})_{js}
    =
    (A_0 \alpha)_{sj}\\
\label{aAjs}
    & =
    \sum_{k,\ell=1}^n
    \theta_{sk\ell}
    E_k
    \alpha_{\ell j}
    =
    -
    \sum_{k,\ell=1}^n
    \theta_{ks\ell}
    E_k
    \alpha_{j \ell},
\end{align}
where the rightmost equality also uses the antisymmetry of the  matrices $\Theta_1, \ldots, \Theta_n$ in (\ref{Thetaell}).
A combination of (\ref{Aajs}) with (\ref{aAjs}) yields
$$
        (A_0 \alpha+\alpha A_0^\rT)_{js}
        =
        \sum_{k=1}^n
        \Big(
        E_k
        \sum_{\ell=1}^n
    (\theta_{jk\ell}
    \alpha_{\ell s}
    -
    \theta_{ks\ell}
    \alpha_{j \ell})
    \Big) = 0,
    \qquad
    j,s = 1, \ldots, n,
$$
where the innermost sum vanishes in view of (\ref{con1Im0}), and hence, for any $E \in \mR^n$,  the matrix $A_0$ satisfies
\begin{equation}
\label{AaaA}
    A_0 \alpha+\alpha A_0^\rT = 0.
\end{equation}
Now, if the matrix $\alpha$ is positive definite, as assumed in (\ref{alfpos}),
then (\ref{AaaA}) allows $A_0$ to be represented as
\begin{equation}
\label{simroot}
    A_0
    =
    \alpha \Ups
    =
    \sqrt{\alpha}
    \sqrt{\alpha} \Ups \sqrt{\alpha}
    \alpha^{-1/2},
    \qquad
    \Ups = -\Ups^\rT \in \mR^{n\x n}.
\end{equation}
Therefore, $A_0$ is isospectral to (as the result of a similarity transformation of) the real antisymmetric matrix $\sqrt{\alpha} \Ups \sqrt{\alpha} $ whose eigenvalues are imaginary (belong to $i\mR$) and symmetric about the origin \cite{HJ_2007}. At least one of these eigenvalues is zero if and only if $\det \Ups = 0$ (which holds with necessity for odd dimensions $n$, including the case of the Pauli matrices  (\ref{X123})).
\end{proof}

Under the assumptions of Theorem~\ref{th:osc} (which are satisfied, in particular, for the Pauli matrices),  the spectrum of the matrix $A_0$ can be  represented as $\{i\omega_k: k = 1, \ldots, n\}$. Here, the eigenfrequencies $\omega_1, \ldots, \omega_n \in \mR$ form a set,  which is symmetric about $0$,  and pertain to the oscillatory (or static if $\omega_k=0$) behavior of the  solutions of the ODE (\ref{Xdot}) for the isolated quantum system. Assuming (\ref{alfpos}) for what follows, the relation
(\ref{simroot}) implies that $\omega_1, \ldots, \omega_n$ coincide with the eigenvalues of the Hermitian matrix
\begin{equation}
\label{VUps}
    -i\sqrt{\alpha}\Ups \sqrt{\alpha} = V\mho V^* \in \mH_n.
\end{equation}
Here,
\begin{equation}
\label{mho}
      \mho := \diag_{1\< k\< n} (\omega_k),
\end{equation}
and
\begin{equation}
\label{Vvv}
    V
    :=
    \begin{bmatrix}
        v_1 & \ldots & v_n
    \end{bmatrix}
    \in \mC^{n\x n}
\end{equation}
is a unitary matrix
whose columns $v_1,\ldots, v_n\in \mC^n$ are the corresponding eigenvectors which satisfy
\begin{equation}
\label{veig}
    -i\sqrt{\alpha}\Ups \sqrt{\alpha} v_k = \omega_k v_k
\end{equation}
along with
\begin{equation}
\label{vsym}
    \omega_j=-\omega_k\
    \Longrightarrow\
    v_j = \overline{v_k},
    \qquad
    j, k = 1, \ldots, n,
    \quad
    j\ne k.
\end{equation}
A combination of (\ref{simroot}) with (\ref{VUps}), (\ref{mho}) implies that the matrix $A_0$ is diagonalisable as
\begin{equation}
\label{AV}
    A_0 = i \Sigma \mho \Sigma^{-1},
    \qquad
    \Sigma := \sqrt{\alpha}V,
\end{equation}
where
\begin{equation}
\label{invSigma}
    \Sigma^{-1} = V^* \alpha^{-1/2}
\end{equation}
due to the unitarity of the matrix (\ref{Vvv}).
Therefore, the dynamic variables of the isolated quantum system (\ref{Xdot}) are related by
$$
    X(t) = \re^{(t-s)A_0}X(s),
    \qquad
    t\> s \> 0,
$$
where the matrix exponential
$$
    \re^{\tau A_0}
    =
    \Sigma\re^{i\tau  \mho} \Sigma^{-1},
$$
associated with (\ref{AV}),  is an oscillatory function of time $\tau$ (with a static component in the presence of zero eigenfrequencies), and so also are the corresponding two-point CCRs
\begin{align}
\nonumber
    [X(t), X(s)^\rT]
    & =
    \re^{\tau A_0} [X(s), X(s)^\rT] \\
\label{XXcommexp}
    & =
    2i \Sigma\re^{i\tau  \mho} \Sigma^{-1} (\Theta \cdot X(s)),
    \qquad
    \tau := t-s \> 0.
\end{align}
The matrix exponential
$$
    \re^{i \tau \mho}
    =
    \diag_{1\< k\< n} (\re^{i\omega_k\tau})
$$
in (\ref{XXcommexp})
represents the isolated system dynamics in terms of auxiliary quantum variables
\begin{equation}
\label{zeta}
    \zeta_k := v_k^* \alpha^{-1/2} X,
    \qquad
    k = 1, \ldots, n,
\end{equation}
which are not necessarily self-adjoint operators and, in accordance with (\ref{invSigma}),  form  a vector
\begin{equation}
\label{zetaX}
    \zeta:= (\zeta_k)_{1\< k \< n}:= \Sigma^{-1}X
\end{equation}
evolving as
\begin{equation}
\label{zetadot}
    \dot{\zeta} = i \mho \zeta.
\end{equation}
The oscillatory part of this evolution can be decomposed into ``planar rotations''  with angular rates $\omega_k>0$ as
\begin{equation}
\label{xietadot}
    \begin{bmatrix}
      \xi_k \\
      \eta_k
    \end{bmatrix}^{^\centerdot}
    =
    -
    \omega_k
    \bJ
    \begin{bmatrix}
      \xi_k \\
      \eta_k
    \end{bmatrix}
\end{equation}
(with the matrix $\bJ$ from (\ref{bJ}))
for pairs $\zeta_j = \zeta_k^\dagger$ of those variables in  (\ref{zeta}), with $j\ne k$,  which correspond to the conjugate eigenvectors $v_j = \overline{v_k}$ in (\ref{veig}) associated with the opposite eigenfrequencies  in (\ref{vsym}). Any such pair $(\zeta_j, \zeta_k)$ enters (\ref{xietadot}) through
\begin{align}
\label{xik}
    \xi_k
    & := \Re \zeta_k = (\Re v_k)^\rT \alpha^{-1/2}X,\\
\label{etak}
    \eta_k
    & := \Im \zeta_k = -(\Im v_k)^\rT \alpha^{-1/2}X,
\end{align}
which are self-adjoint quantum variables. The remaining quantum variables $\zeta_k$  in
(\ref{zeta}), (\ref{zetaX}), with zero eigenfrequencies $\omega_k=0$,  are static in the sense that $\dot{\zeta}_k = 0$ according to (\ref{zetadot}).

The largest period of nontrivial oscillatory modes in (\ref{xietadot}) (specified by positive eigenfrequencies) is given by
\begin{equation}
\label{T*}
  T
  :=
  2\pi\big/ \min \{\omega_k:\ \omega_k >0,\ k = 1, \ldots, n\}.
\end{equation}
Since,  under the condition (\ref{alfpos}),  the matrix $A_0$ is not Hurwitz, then,  in view of (\ref{At}), the Hurwitz property of $A$ in (\ref{A}) can only come from nonzero coupling parameters $M$, $N$.

If the matrix $A$ is Hurwitz, then  the quantum decoherence of the system in the invariant state can be measured by a ``typical'' time constant of decay in the averaged two-point CCR matrix (\ref{EXXcommts}):
\begin{equation}
\label{taud}
    \tau_*
    :=
    \inf
    \Big\{
        \tau>0:\
        \|\re^{\tau A}(\Theta\cdot \mu_*)\|_\rF \< \frac{1}{\re}\|\Theta\cdot \mu_*\|_\rF
    \Big\}.
\end{equation}
A different norm can also be employed instead of the Frobenius norm  $\|\cdot\|$ which takes the form $\|\Theta\cdot \mu_*\|_\rF = \sqrt{-\Tr ((\Theta\cdot \mu_*)^2)}$ in application to the antisymmetric matrix $\Theta\cdot \mu_*$ which results from averaging the one-point CCRs (\ref{XCCRTheta}) over the invariant state: $\bE [X,X^\rT] = 2i\Theta\cdot \mu_*$.

The decoherence time (\ref{taud}) admits an upper bound involving algebraic Lyapunov inequalities as follows. For any
\begin{equation}
\label{lam1}
    0< \lambda < -\sigma(A),
\end{equation}
with $\sigma(A)$ given by (\ref{stab}), there exists a positive definite matrix $G = G^\rT \in \mR^{n\x n}$ satisfying
\begin{equation}
\label{lam2}
    AG + G A^\rT
    \prec -  2\lambda G.
\end{equation}
Since the second inequality in (\ref{lam1}) implies the Hurwitz property for $A+\lambda I_n$, any such matrix $G$ can be represented as the solution
\begin{equation}
\label{K}
    G
    =
    \int_{\mR_+}
    \re^{2\lambda t}
    \re^{tA}
    K
    \re^{tA^\rT}
    \rd t
\end{equation}
of the algebraic Lyapunov equation
$$
    (A+\lambda I_n)G + G (A+\lambda I_n)^\rT + K = 0
$$
with an arbitrary positive definite matrix $K = K^\rT \in \mR^{n\x n}$.
From (\ref{lam2}), it follows that   the matrix
$$
    E_\tau:=
    G^{-1/2}\re^{\tau A}\sqrt{G}
$$
is a contraction in the sense that its operator matrix norm satisfies
$
    \|E_\tau \| \< \re^{-\lambda\tau},
$
or equivalently,
\begin{equation}
\label{contr}
    E_\tau^\rT E_\tau \preccurlyeq \re^{-2\lambda\tau} I_n,
        \qquad
    \tau> 0
\end{equation}
(a similar argument is used, for example, in \cite[Proof of Theorem 6 on p. 122]{VPJ_2018a}).
In combination with
$$
    \re^{\tau A} = \sqrt{G}  E_\tau G^{-1/2},
$$
the inequality (\ref{contr}) leads to
\begin{align}
\nonumber
    \|\re^{\tau A}(\Theta\cdot \mu_*)\|_\rF^2
    & =
    \|\sqrt{G}  E_\tau G^{-1/2}(\Theta\cdot \mu_*)\|_\rF^2\\
\nonumber
    & =
    \Tr ((G^{-1/2}(\Theta\cdot \mu_*) )^\rT
    E_\tau^\rT
    G
    E_\tau
    G^{-1/2}(\Theta\cdot \mu_*))\\
\nonumber
    &
    \<
    \lambda_{\max}(G)
    \Tr ((G^{-1/2}(\Theta\cdot \mu_*) )^\rT
    E_\tau^\rT
    E_\tau
    G^{-1/2}(\Theta\cdot \mu_*))\\
\nonumber
    &
    \<
    \|G\|
        \|E_\tau\|^2
    \Tr ((G^{-1/2}(\Theta\cdot \mu_*) )^\rT
    G^{-1/2}(\Theta\cdot \mu_*))\\
\label{contr1}
    & \<
    \re^{-2\lambda\tau}
    \|G\|
    \|G^{-1/2}(\Theta\cdot \mu_*)\|_\rF^2,
\end{align}
where $\lambda_{\max}(\cdot)$ is the largest eigenvalue of a matrix with a real spectrum.
Therefore, (\ref{contr1}) yields
$$    \|\re^{\tau A}(\Theta\cdot \mu_*)\|_\rF \<     \re^{-\lambda\tau}
    \sqrt{\|G\|}
    \|G^{-1/2}(\Theta\cdot \mu_*)\|_\rF,
$$
so that the decoherence time (\ref{taud}) admits an upper bound
\begin{equation}
\label{tau*max}
    \tau_*
    \<
    \frac{1}{\lambda}
    \Big(
        1
        +
        \ln
        \Big(
            \frac{1}
            {\|(\Theta\cdot \mu_*)\|_\rF}
            \sqrt{\|G\|}\|G^{-1/2}(\Theta\cdot \mu_*)\|_\rF
        \Big)
    \Big).
\end{equation}
In view of the parameterisation (\ref{K}) for the matrix $G$, the bound (\ref{tau*max}) can be tightened by minimising its right-hand side over the pairs $(\lambda,K)$ such that $\lambda$  satisfies (\ref{lam1}) and $K\succ 0$ is normalised, for example, as $\Tr K = 1$. The normalisation does not affect the resulting minimum value since the map $K \mapsto G$ is linear and the quantity $\sqrt{\|G\|}\|G^{-1/2}(\Theta\cdot \mu_*)\|_\rF$ in (\ref{tau*max}) is invariant under the scaling transformation $G\mapsto s G$ for any $s>0$.

%
%
%

\section{Asymptotic decoherence estimates in a weak-coupling formulation}
\label{sec:asy}

With the decoherence part $\wt{A}$ of the matrix (\ref{A}) in (\ref{At}) depending quadratically on the coupling parameters $M$, $N$, this homogeneity can be exploited in a weak-coupling formulation
\begin{equation}
\label{MNeps}
    M_\eps := \eps \sM,
    \qquad
    N_\eps := \eps \sN.
\end{equation}
Here, $\eps \> 0$ is a small scaling factor which quantifies the coupling strength, whereas $\sM \in \mR^{m\x n}$, $\sN \in \mR^m$  specify the coupling ``shape'' in (\ref{H_LM}). The matrices (\ref{A}), (\ref{BX}), (\ref{At}) acquire dependence on $\eps$ as
\begin{equation}
\label{ABeps}
    A_\eps
     := A_0 + \wt{A}_\eps,
    \qquad
     \wt{A}_\eps
      = \eps^2 \sA,
     \qquad
     B_\eps(X)
      :=
      \eps \sB(X),
      \qquad
      \sB(X)
      :=
      2(\Theta \cdot X)\sM^\rT,
\end{equation}
where the matrix $A_0$ is given by (\ref{A0}), and
\begin{equation}
\label{At1}
    \sA
    =
        2
        \Theta \diam (\sM^\rT J\sN)
    +
    2
    \sum_{\ell = 1}^n
    \Theta_\ell
    \sM^\rT
    (
        \sM\theta_{\ell\bullet \bullet}
        +
        J \sM\Re \beta_{\ell\bullet \bullet}
    ).
\end{equation}
The following theorem is concerned with the asymptotic behaviour of the spectrum of the matrix $A_\eps$.

\begin{theorem}
\label{th:asy}
In addition to the assumptions of Theorem~\ref{th:osc}, suppose  the eigenfrequencies of the matrix $A_0$ from (\ref{A0}) in (\ref{simroot})--(\ref{AV})  are   pairwise different:
\begin{equation}
\label{omdiff}
  \omega_j \ne \omega_k,
    \qquad
    j, k = 1, \ldots, n,
    \quad
    j\ne k.
\end{equation}
Then for all sufficiently small values of $\eps$, the eigenvalues $\lambda_1(\eps), \ldots, \lambda_n(\eps)$  of the matrix $A_\eps$ in (\ref{ABeps}) are also pairwise different and, being appropriately numbered, behave asymptotically as
\begin{equation}
\label{eigasy}
  \lambda_k(\eps)
  =
  i\omega_k + \eps^2\nu_k + o(\eps^2),
  \qquad
  {\rm as}\
  \eps \to 0+.
\end{equation}
Here,
\begin{equation}
\label{nu}
  \nu_k:= v_k^*\alpha^{-1/2} \sA \sqrt{\alpha} v_k,
  \qquad
  k = 1, \ldots, n,
\end{equation}
are complex-valued quantities involving
the orthonormal eigenvectors $v_k$ from (\ref{veig}) and the matrices $\alpha$, $\sA$ in (\ref{XXX1}), (\ref{At1}). \hfill$\square$
\end{theorem}
\begin{proof}
The property that the eigenvalues $\lambda_1(\eps), \ldots, \lambda_n(\eps)$ of $A_\eps$ in (\ref{ABeps})  are pairwise different for all sufficiently small $\eps$ (that is, $\eps \in [0, \delta)$ for some $\delta>0$) follows from (\ref{omdiff}) and  the Gershgorin circle  theorem \cite{H_2008,HJ_2007}. This also allows them to be numbered in such a way that each eigenvalue  $\lambda_k(\eps)$ inherits the infinite differentiability from $A_\eps$ over  $\eps \in [0, \delta)$. In view of the similarity transformation in (\ref{AV}),  the matrix $A_\eps$ is isospectral to
\begin{align}
\nonumber
    \Sigma^{-1}A_\eps \Sigma
    & =
    \Sigma^{-1}A_0 \Sigma  + \eps^2 \Sigma ^{-1}\sA\Sigma  \\
\label{SS}
    & =
    i \mho  + \eps^2 V^*\alpha^{-1/2} \sA \sqrt{\alpha} V,
\end{align}
where use is made of (\ref{invSigma}). Since the matrix $\mho$ in (\ref{mho}) is diagonal with pairwise different  diagonal entries $\omega_k$, then,  by the spectrum perturbation results (see, for example, \cite{M_1985} and references therein) applied to (\ref{SS}),  the diagonal entries (\ref{nu}) of the matrix $V^*\alpha^{-1/2} \sA \sqrt{\alpha} V$
specify the coefficients of the linear terms in the asymptotic expansions of $\lambda_k(\eps)$ over the powers of $\eps^2$ in (\ref{eigasy}).
\end{proof}

The following theorem is a corollary of Theorem~\ref{th:asy} and provides stability conditions  for the weakly coupled system.

\begin{theorem}
\label{th:stab}
Suppose the assumptions of Theorem~\ref{th:asy} are satisfied. Then the fulfillment of the inequalities
\begin{equation}
\label{nuneg}
  \Re \nu_k < 0,
  \qquad
  k=1, \ldots, n,
\end{equation}
is sufficient (and in nonstrict form, necessary) for the matrix $A_\eps$ in (\ref{ABeps}) to be  Hurwitz for all $\eps> 0$ small enough. \hfill$\square$
\end{theorem}
\begin{proof}
It follows from (\ref{eigasy}) that
\begin{equation}
\label{Reeigasy}
  \Re \lambda_k(\eps)
  =
  \eps^2 \Re \nu_k   + o(\eps^2),
  \qquad
  {\rm as}\
  \eps \to 0+.
\end{equation}
Hence, the fulfillment of (\ref{nuneg})
is sufficient for $A_\eps$ to be Hurwitz for all $\eps> 0$ small enough, while the nonstrict inequalities
\begin{equation*}
\label{ommu2}
    \Re\nu_k = \lim_{\eps\to 0} \frac{\Re \lambda_k(\eps)}{\eps^2}\< 0,
    \qquad
    k=1, \ldots, n,
\end{equation*}
provide a necessary condition for the Hurwitz property.
\end{proof}

The quantities $\nu_1, \ldots, \nu_n$ in (\ref{nu}) are symmetric about the real axis in the complex plane. Indeed,
\begin{equation}
\label{nubar}
  \overline{\nu_k}
  =
  (\overline{v_k})^*  \alpha^{-1/2} \sA \sqrt{\alpha} \overline{v_k},
  \qquad
  k = 1, \ldots, n,
\end{equation}
due to the matrix $\alpha^{-1/2} \sA \sqrt{\alpha} $ being real. Therefore, since the eigenvectors from (\ref{veig}) with opposite eigenfrequencies are related by the complex  conjugation in (\ref{vsym}), then (\ref{nubar}) leads to
\begin{equation*}
\label{nusym}
    \omega_j=-\omega_k
    \
    \Longrightarrow\
    \nu_j = \overline{\nu_k}
        \
    \Longrightarrow\
    \Re \nu_j = \Re \nu_k,
    \qquad
    j, k = 1, \ldots, n,
    \quad
    j\ne k.
\end{equation*}
Due to this symmetry, the number of independent stability  conditions in (\ref{nuneg}) is essentially  twice smaller (up to a zero eigenfrequency which is necessarily present for odd dimensions $n$). Also note that
\begin{equation}
\label{Renu}
  \Re \nu_k
  = v_k^*\bS(\alpha^{-1/2} \sA \sqrt{\alpha}) v_k,
  \qquad
  k = 1, \ldots, n,
\end{equation}
in view of (\ref{nu}),
depend on the matrix $\alpha^{-1/2} \sA \sqrt{\alpha}$
only through its symmetric part
$$
    \bS(\alpha^{-1/2} \sA \sqrt{\alpha})
    =
    \frac{1}{2}
    (\alpha^{-1/2} \sA \sqrt{\alpha} + \sqrt{\alpha}\sA^\rT \alpha^{-1/2} ),
$$
where $\bS(K):= \frac{1}{2}(K+K^\rT)$ is the symmetrizer of square matrices.
Under the conditions of Theorem~\ref{th:stab}, it follows from (\ref{Reeigasy}) that the Lyapunov exponent (\ref{stab}) satisfies
\begin{equation}
\label{lead}
    \sigma(A_\eps)
    =
    \eps^2 \max_{1\< k \< n}\Re \nu_k + o(\eps^2),
  \qquad
  {\rm as}\
  \eps \to 0+.
\end{equation}
The leading term $\eps^2 \max_{1\< k \< n}\Re \nu_k$ in (\ref{lead}), taken with the opposite sign in view of (\ref{nuneg}),  provides an asymptotically accurate approximation for the right-hand side of (\ref{lam1}) and also an estimate
\begin{equation}
\label{tauhat}
    \wh{\tau}
    :=
    \frac{\eps^{-2}}
    {\max_{1\< k \< n}|\Re \nu_k|}
\end{equation}
for a decay time, which is different from yet related to (\ref{taud}). The decoherence time estimate (\ref{tauhat}) allows the coupling strength $\eps$ to be chosen so as to enable all the isolated oscillatory modes in (\ref{xietadot})--(\ref{etak}) to manifest themselves in a large number of cycles before the decay sets in. This requirement takes the form $\wh{\tau}\gg T$ in terms of the period (\ref{T*}) and leads to
\begin{equation}
\label{epsmax}
    \eps
    \ll
    \sqrt{
    \frac
    {\min\{\omega_k:\ \omega_k>0, \ k = 1, \ldots, n\}}
    {2\pi\min_{1\< k \< n}|\Re \nu_k|}}
    =:
    \wh{\eps}
\end{equation}
as an asymptotic threshold on the coupling strength for the system to preserve the quantum dynamic features of its isolated counterpart.
Since
\begin{equation}
\label{mumax}
    \wh{\eps}
    \>
    1\big/ \sqrt{2\pi \max\{|\Re \nu_k|/\omega_k:\ \omega_k >0,\ k = 1, \ldots, n\}}
    =:
    \wt{\eps},
\end{equation}
in view of
\begin{align*}
    \min_{1\< k \< n}|\Re \nu_k|
    \< &
    \min\{\omega_k:\ \omega_k >0,\ k = 1, \ldots, n\}\\
    & \x \max\{|\Re \nu_k|/\omega_k:\ \omega_k >0,\ k = 1, \ldots, n\},
\end{align*}
the right-hand side of (\ref{mumax}) provides a more stringent threshold on $\eps$. For example, in the case of three system variables $n=3$ (such as the Pauli matrices (\ref{X123})), with the eigenfrequencies $\omega_1=-\omega_2>0$ and $\omega_3 = 0$, both thresholds in (\ref{epsmax}), (\ref{mumax})
are the same and reduce to
\begin{equation}
\label{eps<<}
    \eps
    \ll
    \sqrt{\frac{\omega_1}{2\pi|\Re \nu_1|}}.
\end{equation}
These asymptotic estimates will be complemented by closed-form expressions in  the three-dimen\-sion\-al case in Section~\ref{sec:Pauli}.

\section{Weak-coupling limit of the invariant state}
\label{sec:inv}

While low decoherence is important for isolating the quantum system from the environment, it conflicts with accelerating the  convergence to the invariant state through enhanced dissipation, which underlies the procedure of quantum state generation   mentioned in Section~\ref{sec:QSS}. The system-field coupling influences not only the convergence rate (\ref{lead}), but also the invariant state itself  whose moments (\ref{mured}) are completely specified by the steady-state mean vector (\ref{mu*}).   We will therefore discuss the asymptotic behaviour of the invariant mean vector
 \begin{equation}
\label{mueps}
  \mu_\eps := -A_\eps^{-1} b_\eps
\end{equation}
as $\eps\to 0+$  in the weak-coupling  framework (\ref{MNeps}). Here, the matrix $A_\eps$ is given by (\ref{ABeps}),  and the vector $b$ in  (\ref{b}) also acquires dependence  on the coupling strength $\eps$ as
\begin{equation}
\label{beps}
    b_\eps
     :=
     \eps^2 \sb,
     \qquad
     \sb:=
    2
    \sum_{\ell = 1}^n
    \Theta_\ell
    \sM^\rT
    J\sM\alpha_{\bullet \ell}.
\end{equation}
If the isolated dynamics matrix $A_0$ in  (\ref{A0}) is nonsingular  (so that $n$ is even), then a combination of (\ref{ABeps}) with (\ref{beps}) leads to $\lim_{\eps\to 0+} \mu_\eps = 0$ for (\ref{mueps}) along with the asymptotic relation
\begin{equation*}
\label{muasy}
    \mu_\eps
    \sim
    -\eps^2 A_0^{-1} \sb \to 0,
    \qquad
    {\rm as}\
    \eps \to 0+.
\end{equation*}
In the odd-dimensional case, the asymptotic behaviour of (\ref{mueps}) is qualitatively different.

\begin{theorem}
\label{th:Plim}
Suppose the assumptions of Theorems~\ref{th:asy}, \ref{th:stab} are satisfied. Also, suppose the dimension $n$ is odd  and one of the eigenfrequencies in (\ref{mho}) for the matrix $A_0$ in (\ref{A0}) is zero:
\begin{equation}
\label{omk0}
    \omega_{k_0} = 0
\end{equation}
for some $k_0 = 1, \ldots, n$.
Then the invariant  mean vector $\mu_\eps$  of the system variables in (\ref{mueps}) has the limit
\begin{equation}
\label{mulim}
      \lim_{\eps \to 0+}
    \mu_\eps
    =
    -\frac{1}{\nu_{k_0}}\sqrt{\alpha}
    v_{k_0}
    v_{k_0}^\rT
    \alpha^{-1/2}
    \sb,
\end{equation}
computed in terms of the structure matrix $\alpha$, the eigendata from (\ref{simroot})--(\ref{AV}), the vector $\sb$ from (\ref{ABeps}) and the quantities $\nu_k$ in  (\ref{nu}).  
\hfill$\square$
\end{theorem}
\begin{proof}
Under the conditions of Theorems~\ref{th:asy}, \ref{th:stab}, there exists a $\delta>0$ such that for all $\eps\in [0, \delta)$, the matrix $A_\eps$ has pairwise different eigenvalues $\lambda_1(\eps), \ldots, \lambda_n(\eps)$ in  (\ref{eigasy}) and  is diagonalisable as
\begin{equation}
\label{SAS}
  \Sigma_{\eps}^{-1}  A_{\eps} \Sigma_\eps
  =
  \diag_{1\< k \< n} (\lambda_k(\eps))
  =:
  \Lambda_\eps
\end{equation}
through a nonsingular matrix  $\Sigma_\eps \in \mC^{n\x n}$ whose columns are the corresponding eigenvectors.
The matrix $\Sigma_\eps$ is determined up to a nonsingular diagonal right factor and  can be chosen so as to inherit from $A_\eps$ the infinite differentiability over $\eps \in [0, \delta)$ and to satisfy
\begin{equation}
\label{SSS}
    \Sigma_0
    =
    \lim_{\eps \to 0+} \Sigma_\eps = \Sigma,
\end{equation}
where $\Sigma$ is the matrix from (\ref{AV}).
Due to (\ref{omk0}),  the $k_0$th eigenvalue of the matrix $A_\eps$  in (\ref{eigasy}) satisfies
$$
    \lambda_{k_0}(\eps) = \eps^2 \nu_{k_0} + o(\eps^2),
    \qquad
    {\rm as}\
    \eps\to 0+,
$$
while the other eigenvalues have nonzero limits
$$
    \lim_{\eps \to 0+}
    \lambda_k(\eps) =i\omega_k \ne 0,
    \qquad
    k = 1, \ldots, n,
    \quad
    k\ne k_0,
$$
and hence,
\begin{equation}
\label{epsLam}
    \eps^2\Lambda_\eps^{-1}
    =
    \diag_{1\< k \< n}
    (\eps^2/\lambda_k(\eps))
    \to
    \diag_{1\< k \< n}
    (\delta_{k_0 k}/\nu_{k_0}),
    \qquad
    {\rm as}\
    \eps \to 0+.
\end{equation}
A combination of
$
    A_\eps^{-1} = \Sigma_{\eps} \Lambda_\eps^{-1} \Sigma_\eps^{-1}
$
from (\ref{SAS}) with (\ref{beps}), (\ref{SSS}), (\ref{epsLam}) leads to the following convergence for (\ref{mueps}):
 \begin{align}
\nonumber
  \mu_\eps
  & =
  -\eps^2\Sigma_\eps \Lambda_\eps^{-1} \Sigma_\eps^{-1} \sb\\
\nonumber
    &
    \to
    -\Sigma     \diag_{1\< k \< n}
    (\delta_{k_0 k}/\nu_{k_0}) \Sigma^{-1} \sb\\
\nonumber
    &
    =
    -\sqrt{\alpha} V\diag_{1\< k \< n}
    (\delta_{k_0 k}/\nu_{k_0}) V^* \alpha^{-1/2} \sb\\
\label{mulim1}
    &
    =
    -\frac{1}{\nu_{k_0}}
    \sqrt{\alpha} v_{k_0}v_{k_0}^*
    \alpha^{-1/2} \sb,
    \qquad
    {\rm as}\
    \eps\to 0+,
\end{align}
where the structure of the matrices $\Sigma$, $\Sigma^{-1}$ in (\ref{AV}), (\ref{invSigma}) is also used.   It now remains to note that the eigenvector $v_{k_0}$ of the matrix $A_0$,  associated with the zero eigenfrequency (\ref{omk0}), is real and hence, so also is the quantity $\nu_{k_0}$ in (\ref{nu}), whereby (\ref{mulim1}) establishes (\ref{mulim}).
\end{proof}

The limit (\ref{mulim})  depends on the coupling shape parameters $\sM$, $\sN$  only through the matrix $\sA$ in (\ref{At1}), which enters the  quantity $\nu_{k_0}$ in (\ref{nu}),   and the vector $\sb$ in (\ref{beps}). Furthermore, this limit is invariant under the scaling transformation $(\sM,\sN) \mapsto (s\sM, s\sN)$ for any $s\in \mR\setminus \{0\}$.

\section{The case of Pauli matrices as initial system variables}\label{sec:Pauli}

Suppose the quantum system has $n=3$ dynamic variables organised initially as the Pauli matrices (\ref{X123}) on the Hilbert space $\fH_0:= \mC^2$:
\begin{equation}
\label{XPauli}
    X_k(0) = \sigma_k,
    \qquad
    k = 1,2,3.
\end{equation}
Substitution of the structure constants from (\ref{alfbet}), (\ref{T123}) into 
(\ref{A0}), (\ref{At}) yields the matrix
\begin{equation}
\label{A0Pauli}
    A_0
    =
    2
    \begin{bmatrix}
      0 & -E_3 & E_2\\
      E_3 & 0 & -E_1\\
      -E_2 & E_1 & 0
    \end{bmatrix},
\end{equation}
which  is the infinitesimal generator of rotations about the vector $E:= (E_k)_{1\< k\<  3} \in \mR^3$ with the angular rate
\begin{equation}
\label{om1E}
    \omega_1 = 2|E|
\end{equation}
(the trivial case of $E=0$ is excluded from consideration),
and the matrix
\begin{equation}
\label{AtPauli}
    \wt{A}
     =
        2
        \Theta \diam (M^\rT JN)
    -
    2
    \Gamma,
    \qquad
    \Gamma
    :=
    -
    \sum_{\ell = 1}^3
    \Theta_\ell
    M^\rT
        M
            \Theta_\ell \succcurlyeq 0,
\end{equation}
where $M \in \mR^{m\x 3}$, $N \in \mR^m$ are  the system-field coupling parameters from (\ref{H_LM}).  In (\ref{A0Pauli}), use is made of
(\ref{Thetadiam}), and (\ref{AtPauli}) employs the invariance $\eps_{jk\ell}=\eps_{k\ell j} = \eps_{\ell jk}$ of the Levi-Civita symbol under cyclic permutations of its indices, whereby the sections (\ref{T123}) of the CCR array $\Theta$ satisfy
$
    \Theta_\ell = \theta_{\ell\bullet \bullet}
$ for all $\ell = 1, 2, 3$.  The matrix $A_0$ in (\ref{A0Pauli}) and the first term $2\Theta \diam  (M^\rT J N)$ in (\ref{AtPauli}) are antisymmetric, while the matrix $\Gamma$ in (\ref{AtPauli})  is symmetric, so that (\ref{A}) satisfies
\begin{equation}
\label{AA}
    A+A^\rT
    =
    A_0+A_0^\rT
    +
    \wt{A}+\wt{A}^\rT
    =
    -4\Gamma.
\end{equation}
As established in \cite[Eq. (6.4)]{VP_2022} using the special structure of (\ref{T123}), the matrix $\Gamma$ in (\ref{AtPauli}) admits the representation
\begin{equation*}
\label{MM}
    \Gamma
    =
    \|M\|_\rF^2 I_3 - M^\rT M,
\end{equation*}
which also implies that $\Gamma \succ 0$ (thus ensuring that $A$ is Hurwitz in view of (\ref{AA}))  whenever its rank satisfies
\begin{equation*}
\label{rankM}
    \mathrm{rank }M \> 2.
\end{equation*}
However, this rank condition is sufficient, but not necessary for the stability of the quantum system with the Pauli matrices (\ref{XPauli}). Since $\alpha=I_3$ for the Pauli matrices, then $\Ups = A_0$ in (\ref{simroot}) and the unitary matrix $V \in \mC^{3\x 3}$ in (\ref{VUps}) consists of the orthonormal eigenvectors of the matrix (\ref{A0Pauli}):
\begin{equation}
\label{Vvvu}
    V =
    \begin{bmatrix}
        v_1 & v_2 & v_3
    \end{bmatrix},
    \qquad
    v_{1,2} :=       \frac{1}{\sqrt{2}}(u_1 \pm iu_2),
    \qquad
    v_3:= \frac{1}{|E|}E.
\end{equation}
Here, $u_1, u_2 \in \mR^3$ form an orthonormal basis in the plane $E^{\bot}$, orthogonal to the energy vector  $E$, so that $v_1$ and $v_2=\overline{v_1}$ are the eigenvectors of $A_0$ which correspond to the nonzero eigenfrequencies $\omega_1$ and $\omega_2=-\omega_1$ in (\ref{om1E}):
$$
    A_0 v_{1,2} = \pm i\omega_1 v_{1,2},
$$
or, equivalently,
$$
    A_0
    \begin{bmatrix}
        u_1 & u_2
    \end{bmatrix}
    =
    \omega_1
    \begin{bmatrix}
        u_1 & u_2
    \end{bmatrix}
    \bJ,
$$
with $\bJ$ given by (\ref{bJ}).
Also,  $v_3$ in (\ref{Vvvu})
is a unit vector in $\mR^3$ associated with the zero eigenfrequency $\omega_3 = 0$ of the matrix (\ref{A0Pauli}), so that
$$
    A_0v_3 = 0.
$$
Accordingly,
$    \begin{bmatrix}
        u_1 & u_2 & v_3
    \end{bmatrix} \in \mR^{3\x 3}
$
is an orthogonal matrix which block-diagonalises the matrix $A_0$. Now, in the weak-coupling formulation (\ref{MNeps}), the quantities (\ref{nu}) take the form
\begin{equation}
\label{nuPauli}
  \nu_k:= v_k^* \sA v_k,
  \qquad
  k = 1, 2, 3,
\end{equation}
where
\begin{equation}
\label{At1Pauli}
    \sA
     =
        2
        \Theta \diam (\sM^\rT J\sN)
    -
    2
    \sGamma,
    \qquad
    \sGamma:=
    -
    \sum_{\ell = 1}^3
    \Theta_\ell
    \sM^\rT
        \sM
            \Theta_\ell,
\end{equation}
are expressed in terms of the coupling shape parameters $\sM \in \mR^{m\x 3}$, $\sN \in \mR^m$ in accordance with (\ref{AtPauli}). Therefore, in view of (\ref{Renu}), the real parts of (\ref{nuPauli}) depend on the matrix $\sA$ only through its symmetric part $\bS(\sA) =-2\sGamma$ in (\ref{At1Pauli}) (with         the antisymmetric part $2
        \Theta \diam (\sM^\rT J\sN)$ being irrelevant):
\begin{equation}
\label{Renu1}
    \Re \nu_1 = \Re \nu_2 = -2(\|u_1\|_{\sGamma}^2+\|u_2\|_{\sGamma}^2),
    \qquad
    \Re \nu_3 = -2 \|v_3\|_{\sGamma}^2.
\end{equation}
The closed-form expressions (\ref{om1E}), (\ref{Renu1}) can be used in the system-field coupling strength threshold (\ref{eps<<}).

\section{Interconnected systems with direct energy coupling}\label{sec:connect}

We will now apply the results of the previous sections to a decoherence control setting for an interconnection of
two quantum systems (interpreted, for example,  as a plant and a controller). These  systems are  endowed  with initial spaces $\fH_0^{(s)}$ and vectors $X^{(s)} := (X_k^{(s)})_{1\< k \< n_s}$ of time-varying self-adjoint dynamic variables $X_k^{(s)}$ with an algebraic structure as in (\ref{XXX}) along with the conditions  (\ref{herm}), (\ref{Imalpha0}):
\begin{equation}
\label{XXX12}
    \Xi_{jk}^{(s)}
    :=
    X_j^{(s)} X_k^{(s)}
    =
    \alpha_{jk}^{(s)}
    +
    \sum_{\ell=1}^{n_s}\beta_{jk\ell}^{(s)} X_\ell^{(s)},
    \qquad
    j,k=1, \ldots, n_s,
    \quad
    s = 1,2,
\end{equation}
so that, similarly to (\ref{XXX1}),
\begin{equation}
\label{Xis}
    \Xi^{(s)}
    :=
    (\Xi_{jk}^{(s)})_{1\< j,k\< n_s}
    =
    X^{(s)}X^{(s)\rT}
    =
    \alpha^{(s)} + \beta^{(s)} \cdot X^{(s)}.
\end{equation}
Here, $\alpha^{(s)} := (\alpha_{jk}^{(s)})_{1\< j,k\< n_s}$ is a real symmetric matrix,  and $\beta^{(s)}:= (\beta_{jk\ell}^{(s)})_{1\< j,k,\ell\< n_s}$   is a complex array with Hermitian sections
\begin{equation}
\label{betells}
    \beta_\ell^{(s)}
    :=
    \beta_{\bullet \bullet \ell}^{(s)}
    :=
    (\beta_{jk\ell}^{(s)})_{1\< j,k\< n_s}
    \in
    \mH_{n_s},
    \qquad
    \ell = 1, \ldots, n_s,
    \quad
    s = 1,2.
\end{equation}
In accordance with (\ref{Theta}), (\ref{Thetaell}),  the individual CCR coefficients
\begin{equation}
\label{CCRs}
    \theta_{jk\ell}^{(s)}
    :=
    \Im \beta_{jk\ell}^{(s)}
    =
    - \theta_{kj\ell}^{(s)}
\end{equation}
form arrays $\Theta^{(s)}:= (\theta_{jk\ell}^{(s)})_{1\< j,k,\ell\< n_s}$ with sections $\Theta_\ell^{(s)}:= (\theta_{jk\ell}^{(s)})_{1\< j,k\< n_s}$ (which are real antisymmetric matrices).
As operators on different spaces $\fH_0^{(s)}$, the initial variables of the two systems commute, and this cross-commutativity is preserved by their subsequent  evolution, so that
\begin{equation}
\label{X12comm}
    [X^{(1)}, X^{(2)\rT}]
    :=
    ([X_j^{(1)}, X_k^{(2)}])_{1\< j\< n_1, 1\< k\< n_2}
    = 0
\end{equation}
holds at any time $t\> 0$. Hence, the self-adjointness of the constituent system variables is  inherited by their pairwise products
\begin{equation}
\label{XXab12}
    \Xi_{jk}^{(12)}
    :=
    X_j^{(1)}
    X_k^{(2)}
    =
    X_k^{(2)}
    X_j^{(1)}
    =:
    \Xi_{kj}^{(21)},
\end{equation}
which form the matrices
\begin{align}
\label{Xi12}
    \Xi^{(12)}
    & :=
    (\Xi_{jk}^{(12)})_{1\< j\< n_1, 1\< k \< n_2}
    =
    X^{(1)}X^{(2)\rT},\\
\label{Xi21}
    \Xi^{(21)}
    & :=
    (\Xi_{jk}^{(21)})_{1\< j\< n_2, 1\< k \< n_1}
    =
    X^{(2)}X^{(1)\rT} = \Xi^{(12)\rT},
\end{align}
with the cross-commutativity (\ref{X12comm}) being essential for  (\ref{XXab12}) and   the rightmost equality in (\ref{Xi21}).
The vectorisation of these matrices leads to the following Kronecker products:
\begin{equation}
\label{X12X21}
    \vec{\Xi}^{(12)}
    =
    X^{(2)}\ox X^{(1)}
    =
    \sP
    X^{(12)},
    \qquad
    X^{(12)}
    :=
    X^{(1)}\ox X^{(2)}
    =
    \vec{\Xi}^{(21)} ,
\end{equation}
which are related by a permutation matrix  $\sP$ of order $n_1n_2$,  defined uniquely  by the condition that $\vec{K} = \sP \vect(K^\rT)$ for any $(n_1\x n_2)$-matrix $K$. The matrix $\sP$ is involutory ($\sP^2 = I_{n_1n_2}$),  and
\begin{equation}
\label{Puv}
    v \ox u
    =
    \vect(uv^\rT)
    =
    \sP
    \vect(vu^\rT)
    =
    \sP
    (u\ox v),
    \qquad
    u\in \mR^{n_1},\
    v \in \mR^{n_2},
\end{equation}
which extends to any vectors $u$, $v$ of $n_1$ and $n_2$ quantum variables, respectively,  satisfying $[u,v^\rT]=0$, with the second equality in (\ref{X12X21}) being a particular case of (\ref{Puv}). In what follows, we will use the augmented set
\begin{equation}
\label{XXXi}
    \{X_j^{(1)},\ X_k^{(2)},\ \Xi_{jk}^{(12)}:\ j = 1, \ldots, n_1,\ k = 1, \ldots, n_2\},
\end{equation}
of
\begin{equation}
\label{nnn}
    n := n_1 + n_2 + n_1n_2 = (n_1+1)(n_2+1)-1
\end{equation}
quantum variables (which are all self-adjoint) and assemble them into a vector as
\begin{equation}
\label{Xvec}
    X:=
    \begin{bmatrix}
      X^{(1)}\\
      X^{(2)}\\
      X^{(12)}
    \end{bmatrix},
\end{equation}
with $X^{(12)}$ given by (\ref{X12X21}); cf. \cite[Example~2]{EMPUJ_2016}. Note that $X$ consists of the component system variables and their pairwise products (\ref{XXab12}).

\begin{lemma}
\label{lem:alfbet}
For the self-adjoint constituent system  variables with the algebraic and commutation structure  (\ref{XXX12})--(\ref{X12comm}), the quantum variables in  (\ref{XXXi}) also have an algebraic structure in the sense that the vector $X$ in (\ref{Xvec}) satisfies (\ref{XXX1}), where the coefficients  are computed in terms of the individual structure constants.   In particular,
\begin{equation}
\label{alfcomp}
    \alpha
    =
    \begin{bmatrix}
      \alpha^{(1)} & 0 & 0\\
      0 & \alpha^{(2)} & 0\\
      0 & 0 & \alpha^{(1)}\ox \alpha^{(2)}
    \end{bmatrix}.
\end{equation}
\hfill$\square$
\end{lemma}
\begin{proof}
In application to the vector $X$ in (\ref{Xvec}), the matrix $\Xi$, defined by the first two equalities in (\ref{XXX1}), takes the form
\begin{equation}
\label{Xicomp}
    \Xi
    :=
    XX^\rT
    =
    \begin{bmatrix}
      \Xi^{(1)} & \Xi^{(12)} & X^{(1)}X^{(12)\rT}\\
      \Xi^{(21)} & \Xi^{(2)} &  X^{(2)}X^{(12)\rT}\\
      X^{(12)}X^{(1)\rT} & X^{(12)}X^{(2)\rT}& X^{(12)}X^{(12)\rT}
    \end{bmatrix}.
\end{equation}
In view of (\ref{Xis}), the matrix $\Xi^{(s)}$ is an affine function of $X^{(s)}$ for each $s = 1,2$.  The matrix $\Xi^{(12)}$ in (\ref{Xi12}) consists of the entries (\ref{XXab12}) of the vector $X^{(12)}$ from (\ref{X12X21}),  and hence, is a linear function of $X^{(12)}$. A linear dependence on $X^{(12)}$ is inherited from $\Xi^{(12)}$ by the matrix $\Xi^{(21)}$ in (\ref{Xi21}). We will now prove that the other five blocks of the matrix $\Xi$  in (\ref{Xicomp}) are also affine functions of $X$, which, due to (\ref{Xi+}),  reduces to establishing an affine dependence on $X$ for the remaining upper off-diagonal blocks $X^{(1)}X^{(12)\rT}$, $X^{(2)}X^{(12)\rT}$ and the diagonal block $X^{(12)}X^{(12)\rT}$. To this end, note that the matrix
\begin{align}
\nonumber
    X^{(1)}X^{(12)\rT}
    & =
    X^{(1)}(X^{(1)\rT}\ox X^{(2)\rT})\\
\nonumber
    & =
    \Xi^{(1)}\ox X^{(2)\rT}    \\
\nonumber
    & =
    (\alpha^{(1)} + \beta^{(1)}\cdot X^{(1)})\ox X^{(2)\rT}    \\
\label{X1X12}
    & =
    \alpha^{(1)} \ox X^{(2)\rT}
    +
    \sum_{\ell=1}^{n_1}
    \beta_\ell^{(1)}
    \ox
    (X_\ell^{(1)}X^{(2)\rT})
\end{align}
is a linear function of $X^{(2)}$, $X^{(12)}$, where (\ref{Xis}) is used.
In a similar fashion, a combination of (\ref{X12X21}) with (\ref{Xis}) leads to
\begin{align}
\nonumber
    X^{(2)}X^{(12)\rT}
    & =
    X^{(2)}(X^{(2)\rT}\ox X^{(1)\rT})\sP^\rT \\
\nonumber
    & =
    (\Xi^{(2)}\ox X^{(1)\rT})\sP^\rT    \\
\nonumber
    & =
    ((\alpha^{(2)} + \beta^{(2)}\cdot X^{(2)})\ox X^{(1)\rT})\sP^\rT    \\
\label{X2X12}
    & =
    \Big(
        \alpha^{(2)} \ox X^{(1)\rT}
    +
    \sum_{\ell=1}^{n_2}
    \beta_\ell^{(2)}
    \ox
    (X_\ell^{(2)}X^{(1)\rT})
    \Big)\sP^\rT,
\end{align}
which is a linear function of $X^{(1)}$, $X^{(12)}$. By using the cross-commutativity (\ref{X12comm}) together with  (\ref{Xis}), it follows that
\begin{align}
\nonumber
    X^{(12)}X^{(12)\rT}
    = &
    (X^{(1)}\ox X^{(2)}) (X^{(1)\rT}\ox X^{(2)\rT})\\
\nonumber
    = &
    \Xi^{(1)}\ox \Xi^{(2)}\\
\nonumber
    = &
    (\alpha^{(1)} + \beta^{(1)}\cdot X^{(1)})
    \ox
    (\alpha^{(2)} + \beta^{(2)}\cdot X^{(2)})\\
\nonumber
    = &
    \alpha^{(1)}
    \ox
    \alpha^{(2)} \\
\nonumber
    & +
    \sum_{j=1}^{n_1}
    (\beta_j^{(1)} \ox \alpha^{(2)})
    X_j^{(1)}
    +
    \sum_{k=1}^{n_2}
    (\alpha^{(1)}
    \ox
    \beta_k^{(2)})
    X_k^{(2)}\\
\label{X12X12}
    & +
    \sum_{j=1}^{n_1}
    \sum_{k=1}^{n_2}
    (\beta_j^{(1)} \ox \beta_k^{(2)})
    \Xi_{jk}^{(12)}
\end{align}
is an affine function of $X$ from (\ref{Xvec}). The relations (\ref{X1X12})--(\ref{X12X12}) complete the verification of the property that the matrix $\Xi$ in (\ref{Xicomp}) is an affine function of $X$, which can therefore be represented by the right-hand side of (\ref{XXX1}). Since the off-diagonal blocks of $\Xi$ are linear functions of $X$, then only its diagonal blocks $\Xi^{(1)}$, $\Xi^{(2)}$, $X^{(12)}X^{(12)\rT}$ contribute to the constant term $\alpha$ in (\ref{XXX1}), which takes the form of (\ref{alfcomp}) and  is a real symmetric matrix (inheriting these properties from $\alpha^{(1)}$, $\alpha^{(2)}$).
\end{proof}

Complementing (\ref{alfcomp}) of Lemma~\ref{lem:alfbet}, the other structure constants for the vector $X$ in (\ref{Xvec}), which form the array $\beta$ in (\ref{XXX1}), can be recovered from the representations (\ref{Xis}), (\ref{X1X12})--(\ref{X12X12}) for the blocks of the matrix $\Xi$ in (\ref{Xicomp}).  For example, (\ref{X12X12}) allows the corresponding diagonal blocks of the sections $\beta_\ell$ to be identified with the Hermitian matrices $\beta_j^{(1)}\ox \alpha^{(2)}$, $\alpha^{(1)}\ox \beta_k^{(2)}$, $\beta_j^{(1)}\ox \beta_k^{(2)}$ involving the sections (\ref{betells}). Their imaginary parts are real antisymmetric matrices $\Theta_j^{(1)}\ox \alpha^{(2)}$, $\alpha^{(1)} \ox \Theta_k^{(2)}$,   $\Theta_j^{(1)}\ox  \Re \beta_k^{(2)} +  \Re \beta_j^{(1)}\ox \Theta_k^{(2)}$, which use (\ref{CCRs}) and specify the CCR coefficients for $X^{(12)}$ as
\begin{align}
\nonumber
    [X^{(12)}, X^{(12)\rT}]
    = &
    2i
    \Big(
    \sum_{j=1}^{n_1}
    (\Theta_j^{(1)} \ox \alpha^{(2)})
    X_j^{(1)}
    +
    \sum_{k=1}^{n_2}
    (\alpha^{(1)}
    \ox
    \Theta_k^{(2)})
    X_k^{(2)}\\
\label{X12CCR}
    & +
    \sum_{j=1}^{n_1}
    \sum_{k=1}^{n_2}
    (\Theta_j^{(1)}\ox  \Re \beta_k^{(2)} +  \Re \beta_j^{(1)}\ox \Theta_k^{(2)})
    \Xi_{jk}^{(12)}
    \Big).
\end{align}

Now, the two systems under consideration are directly coupled through the total Hamiltonian
\begin{equation}
\label{HHH}
    H
    =
    H_1 + H_2
    +
    H_{12},
\end{equation}
which is the sum of the individual Hamiltonians
\begin{equation}
\label{Hs}
    H_s
    :=
    E^{(s)\rT} X^{(s)}
    =
    \sum_{k=1}^{n_s}
    E_k^{(s)}X_k^{(s)},
    \qquad
    s = 1, 2,
\end{equation}
of the systems with the energy vectors $E^{(s)}:= (E_k^{(s)})_{1\< k \< n_s} \in \mR^{n_s}$ and the interaction Hamiltonian
\begin{equation}
\label{H12}
    H_{12}
    :=
    \sum_{j=1}^{n_1}
    \sum_{k=1}^{n_2}
    E_{jk}^{(12)}
    \Xi_{jk}^{(12)}
    =
    X^{(1)\rT} E^{(12)} X^{(2)}
    =
    \vec{E}^{(21)\rT}
    X^{(12)}
\end{equation}
parameterised by a matrix $E^{(12)}:= (E_{jk}^{(12)})_{1\< j\< n_1, 1\< k\< n_2} \in \mR^{n_1\x n_2}$. Here, $\vec{E}^{(21)} \in \mR^{n_1n_2}$ is the vectorisation of the matrix
$    E^{(21)}:= E^{(12)\rT} \in \mR^{n_2\x n_1}$, and use is made of  $X^{(12)}$ from  (\ref{X12X21}). Therefore,
the Hamiltonian $H$, defined by (\ref{HHH})--(\ref{H12}),   is a linear function of $X$ from (\ref{Xvec}),
which can be  described by the first equality in (\ref{H_LM}) with an augmented energy vector $E$:
\begin{equation}
\label{HEX}
    H
    =
    E^\rT X,
    \qquad
    E:=
    \begin{bmatrix}
      E^{(1)}\\
      E^{(2)}\\
      \vec{E}^{(21)}
    \end{bmatrix}
    \in \mR^n.
\end{equation}

In addition to the direct energy coupling, the systems interact with external bosonic fields modelled by
multichannel quantum Wiener processes  $W^{(s)}$ of even dimensions $m_s$ on symmetric Fock spaces $\fF_s$, $s=1,2$; see
Fig.~\ref{fig:system}.
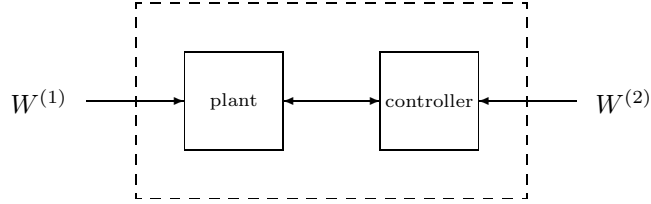
\begin{figure}[htbp]
\centering
\unitlength=1.3mm
\linethickness{0.2pt}
\begin{picture}(50.00,23.00)
    \put(5,0){\dashbox(40,20)[cc]{}}
    \put(10,5){\framebox(10,10)[cc]{\scriptsize plant}}
    \put(30,5){\framebox(10,10)[cc]{\scriptsize controller}}
    \put(0,10){\vector(1,0){10}}
    \put(50,10){\vector(-1,0){10}}


    \put(20,10){\vector(1,0){10}}
    \put(30,10){\vector(-1,0){10}}
    \put(-2,10){\makebox(0,0)[rc]{$W^{(1)}$}}
    \put(52,10){\makebox(0,0)[lc]{$W^{(2)}$}}
\end{picture}
\caption{
    An interconnection of two quantum systems (a plant and a controller), which have external input quantum Wiener processes  $W^{(1)}$, $W^{(2)}$  and interact with each other through a direct  energy coupling shown as a double arrow.
}
\label{fig:system}
\end{figure}
These processes form an augmented quantum Wiener process
$    W
    :=
    (W^{(s)})_{s=1,2}
$
of dimension $
    m:= m_1+m_2$ on the composite Fock space $\fF:= \fF_1\ox \fF_2$ and, similarly to (\ref{dWdW_Omega_J_bJ}), have the Ito tables
\begin{equation}
\label{wwww}
    \rd W^{(s)}\rd W^{(s)\rT} = \Omega_s \rd t,
    \qquad
    \rd W \rd W^{\rT} = \Omega \rd t   ,
    \qquad
    s = 1, 2,
\end{equation}
with the quantum Ito matrices $\Omega_s$, $\Omega$ given by
\begin{align}
\nonumber
    \Omega_s
    & := I_{m_s} + iJ_s,
    \qquad
    J_s:= \bJ \ox I_{m_s/2},\\
\label{Om12}
    \Omega
    & :=
    \blockdiag_{s=1,2}(\Omega_s)
    =
    I_m + iJ,
    \qquad
    J:=
    \blockdiag_{s=1,2}(J_s)    .
\end{align}
The fully quantum system interconnection in Fig.~\ref{fig:system}, driven by the external fields,  is endowed with the system-field space $\fH:= \fH_0 \ox \fF$, where  $\fH_0:= \fH_1 \ox \fH_2$ is the initial space of the composite system.
It is assumed that the operators $L_1^{(s)}, \ldots, L_{m_s}^{(s)}$ of coupling of the $s$th constituent system to the corresponding input field $W^{(s)}$ are not affected by the direct coupling between the systems and, similarly to the second equality in  (\ref{H_LM}),   are affine functions of the dynamic  variables of the $s$th system:
\begin{equation}
\label{LMNs}
    L^{(s)}
    :=
    (L_k^{(s)})_{1\< k \< m_s}
     =
     M^{(s)}X^{(s)} + N^{(s)},
     \qquad
     s = 1,2,
\end{equation}
where $M^{(s)} \in \mR^{m_s\x n_s}$, $N^{(s)} \in \mR^{m_s}$. Accordingly, the vector of $m$ operators of coupling of the composite system to the augmented quantum Wiener process $W$ is related to (\ref{Xvec}) by
\begin{equation}
\label{LMN}
    L
    :=
    \begin{bmatrix}
      L^{(1)}\\
      L^{(2)}
    \end{bmatrix}
    =
    MX + N,
    \qquad
    M:=
    \begin{bmatrix}
      M^{(1)} & 0 & 0\\
      0 & M^{(2)} & 0
    \end{bmatrix},
    \qquad
    N
    :=
    \begin{bmatrix}
      N^{(1)}\\
      N^{(2)}
    \end{bmatrix}    .
\end{equation}
The relations (\ref{HEX}), (\ref{LMNs}), (\ref{LMN}) specify the energetics of the quantum system interconnection and its interaction with the external fields in Fig.~\ref{fig:system}. Together with the algebraic and commutatation structure of the system and input field  variables in (\ref{XXX12})--(\ref{X12comm}), (\ref{wwww}), (\ref{Om12}),  these relations make Theorem~\ref{th:QSDE} applicable to the composite system, thus leading to a  set of quasilinear QSDEs provided below. In the absence of direct coupling between the constituent systems, Theorem~\ref{th:QSDE} would also apply to  each of them, resulting in a quasilinear QSDE of type (\ref{dX1}) with the following matrix $A^{(s)} \in \mR^{n_s\x n_s}$, vector $b^{(s)} \in \mR^{n_s}$ and $(n_s\x m_s)$-matrix $B^{(s)}(X^{(s)})$ of self-adjoint operators (depending linearly on $X^{(s)}$):
\begin{align}
\nonumber
    A^{(s)}
     := &
        2
        \Theta^{(s)} \diam (E^{(s)} + M^{(s)\rT} J_sN^{(s)})\\
\label{As}
    & +
    2
    \sum_{\ell = 1}^{n_s}
    \Theta_\ell^{(s)}
    M^{(s)\rT}
    (
        M^{(s)}\theta_{\ell\bullet \bullet}^{(s)}
        +
        J_s M^{(s)}\Re \beta_{\ell\bullet \bullet}^{(s)}
    ),\\
\label{bs}
    b^{(s)}
    := &
    2
    \sum_{\ell = 1}^{n_s}
    \Theta_\ell^{(s)}
    M^{(s)\rT}
    J_sM^{(s)}\alpha_{\bullet \ell}^{(s)},\\
\label{BXs}
    B^{(s)}(X^{(s)})
     := &
     2(\Theta^{(s)} \cdot X^{(s)})M^{(s)\rT},
     \qquad
     s = 1,2,
\end{align}
which  are obtained by applying (\ref{A})--(\ref{BX}) to the uncoupled systems.
We will also use auxiliary matrices
\begin{align}
\label{F1}
    F_1
    & :=
    2
        \begin{bmatrix}
          \Theta_1^{(1)}E^{(12)}
          & \ldots &
          \Theta_{n_1}^{(1)}E^{(12)}
        \end{bmatrix}
        =
    2\Delta_1
    (I_{n_1}\ox E^{(12)}),\\
\label{F2}
    F_2
    & :=
    2
        \begin{bmatrix}
          \Theta_1^{(2)}E^{(21)}
          & \ldots &
          \Theta_{n_2}^{(2)}E^{(21)}
        \end{bmatrix}
        \sP
    =
        2\Delta_2
    (I_{n_2}\ox E^{(21)})
    \sP,
\end{align}
where the matrices
\begin{equation}
\label{Deltas}
    \Delta_s
    :=
        \begin{bmatrix}
          \Theta_1^{(s)} & \ldots &  \Theta_{n_s}^{(s)}
        \end{bmatrix},
        \qquad
        s = 1,2,
\end{equation}
are formed from the sections of the CCR arrays $\Theta^{(s)}$ in (\ref{CCRs}), similarly to (\ref{Delta}).
Also, use will be made of the matrices
\begin{align}
\label{G1}
      G_1
      & :=
    2
    \begin{bmatrix}
        (\Theta_1^{(1)} \ox \alpha^{(2)})\vec{E}^{(21)}
        &
        \ldots
        &
        (\Theta_{n_1}^{(1)} \ox \alpha^{(2)})\vec{E}^{(21)}
    \end{bmatrix},\\
\label{G2}
      G_2
      & :=
    2
    \begin{bmatrix}
        (\alpha^{(1)}\ox\Theta_1^{(2)})\vec{E}^{(21)}
        &
        \ldots
        &
        (\alpha^{(1)}\ox\Theta_{n_2}^{(2)})\vec{E}^{(21)}
    \end{bmatrix},    \\
\label{G12}
      G_{12}
      & :=
    2
    \begin{bmatrix}
        \gamma_1
        &
        \ldots
        &
        \gamma_{n_1}
    \end{bmatrix},
    \qquad
    \gamma_j
    :=
    \begin{bmatrix}
        \gamma_{j1}
        &
        \ldots
        &
        \gamma_{jn_2}
    \end{bmatrix},
    \qquad
    j = 1, \ldots, n_1,
\end{align}
where the blocks $\gamma_j \in \mR^{n_1n_2\x n_2}$ of the matrix $G_{12}$ consist of the following columns:
\begin{equation}
\label{Rebet}
  \gamma_{jk}
  :=
    (\Theta_j^{(1)}\ox  \Re \beta_k^{(2)} +  \Re \beta_j^{(1)}\ox \Theta_k^{(2)})
    \vec{E}^{(21)} ,
    \qquad
        k = 1, \ldots, n_2,
\end{equation}
which involve the sections (\ref{betells}) along with their real and imaginary parts and the direct coupling matrix from (\ref{H12}).

\begin{theorem}
\label{th:comp}
The direct energy coupling of the quantum systems with the algebraic structure (\ref{XXX12})--(\ref{X12comm}), the augmented system variables in (\ref{Xvec}), the external quantum fields in (\ref{wwww}), (\ref{Om12}), the  Hamiltonian (\ref{HEX}) and the system-field coupling  (\ref{LMNs}), (\ref{LMN}),  is governed by the QSDE (\ref{dX1}), where
\begin{align}
\label{Acomp}
    A
    & =
    \begin{bmatrix}
      A^{(1)} & 0 & F_1\\
      0 & A^{(2)} & F_2\\
    \sP (b^{(2)}\ox I_{n_1}) + G_1
    &
    b^{(1)}\ox I_{n_2} + G_2
    &
    A^{(1)} \op A^{(2)} + G_{12}
    \end{bmatrix},\\
\label{bcomp}
    b
    & =
    \begin{bmatrix}
      b^{(1)}\\
      b^{(2)}\\
      0
    \end{bmatrix},\\
\label{Bcomp}
    B(X)
    & =
    \begin{bmatrix}
      B^{(1)}(X^{(1)}) & 0\\
      0 & B^{(2)}(X^{(2)})\\
      \fB_1(X^{(12)}) & \fB_2(X^{(12)})
    \end{bmatrix}.
\end{align}
Here,  (\ref{As})--(\ref{BXs}) are used along with the matrices from (\ref{F1}), (\ref{F2}), (\ref{G1})--(\ref{G12}). Also, $\fB_s$ are linear maps acting on an $(n_1n_2)$-dimensional vector $\xi$ and producing $(n_1n_2\x m_s)$-matrices as
\begin{align}
\label{fB1def}
    \fB_1(\xi)
    & :=
    2
    \begin{bmatrix}
      (\theta_{\bullet 1 \bullet}^{(1)}\ox I_{n_2})\xi &
      \ldots &
      (\theta_{\bullet n_1 \bullet}^{(1)}\ox I_{n_2}) \xi
    \end{bmatrix}
    M^{(1)\rT},\\
\label{fB2def}
    \fB_2(\xi)
    & :=
    2
    \begin{bmatrix}
      (I_{n_1}\ox \theta_{\bullet 1 \bullet}^{(2)})\xi &
      \ldots &
      (I_{n_1}\ox \theta_{\bullet n_2 \bullet}^{(2)}) \xi
    \end{bmatrix}
    M^{(2)\rT},
\end{align}
using the system-field coupling matrices and appropriate sections of the CCR arrays.
\hfill$\square$
\end{theorem}
\begin{proof}
Since the quantum Ito matrix $\Omega$ in (\ref{Om12}) is block-diagonal and the vector $L$ in (\ref{LMN}) consists of the subvectors $L^{(s)}$, the decoherence superoperator $\cD$ in (\ref{cD}) can be  decomposed into the sum
\begin{equation}
\label{DDD}
    \cD = \cD_1 + \cD_2
\end{equation}
of the decoherence superoperators associated with the constituent systems and acting as
\begin{equation}
\label{cDs}
    \cD_s(\xi)
    :=
     \frac{1}{2}
    ([L^{(s)\rT},\xi]\Omega_s L^{(s)} + L^{(s)\rT} \Omega_s [\xi,L^{(s)}]),
    \qquad
    s = 1,2.
\end{equation}
A similar yet different decomposition holds for the generator $\cG$ of the composite system in (\ref{cG}):
\begin{equation}
\label{GGG}
  \cG  = \cG_1 + \cG_2 + i[H_{12}, \cdot],
\end{equation}
where use is made of (\ref{HHH}), (\ref{DDD}), and
\begin{equation}
\label{cGs}
\cG_s
   := i[H_s,\cdot]
     +
     \cD_s,
     \qquad
     s = 1,2
\end{equation}
are the individual generators.
From the cross-commutativity (\ref{X12comm}) and the fact that $L^{(s)}$ in (\ref{LMNs}) depends only on $X^{(s)}$, it follows  that $[X^{(3-s)}, L^{(s)\rT}] = 0$ for any $s=1,2$. Hence, the individual decoherence superoperators in (\ref{cDs}) satisfy $    \cD_s(X^{(3-s)}) = 0$, and so also do the generators in (\ref{cGs}) since $H_s$ in (\ref{Hs}) depends only on $X^{(s)}$:
\begin{equation}
\label{cG12}
    \cG_s(X^{(3-s)})
    =
    i[H_s, X^{(3-s)}] + \cD_s(X^{(3-s)})
    =
    0,
    \qquad
    s = 1,2.
\end{equation}
By a similar reasoning,
\begin{equation}
\label{BBB}
    [X^{(1)}, L^\rT]
    =
    \begin{bmatrix}
      [X^{(1)}, L^{(1)\rT}] & 0
    \end{bmatrix},
    \qquad
    [X^{(2)}, L^\rT]
    =
    \begin{bmatrix}
      0 & [X^{(2)}, L^{(2)\rT}]
    \end{bmatrix}.
\end{equation}
A combination of (\ref{GGG})--(\ref{BBB}) with (\ref{dX}) and the fact that $W$ consists of $W^{(1)}$, $W^{(2)}$ as subvectors  leads to the QSDEs
\begin{align}
\nonumber
  \rd X^{(s)}
  & =
  \cG(X^{(s)})\rd t - i[X^{(s)}, L^\rT]\rd W \\
\nonumber
    & =
    (\cG_s(X^{(s)}) + i[H_{12}, X^{(s)}])\rd t - i[X^{(s)}, L^{(s)\rT}]\rd W^{(s)}\\
\label{dXs}
    & =
    (A^{(s)}X^{(s)} + b^{(s)}+i[H_{12}, X^{(s)}]) \rd t + B^{(s)}(X^{(s)})\rd W^{(s)},
    \qquad
    s = 1,2,
\end{align}
where use is made of (\ref{As})--(\ref{BXs}) which, by Theorem~\ref{th:QSDE},  specify the individual generators and the dispersion matrices applied to the corresponding system variables:
$$
    \cG_s(X^{(s)})
    =
    A^{(s)}X^{(s)} + b^{(s)},
    \qquad
    -i[X^{(s)}, L^{(s)\rT}] = B^{(s)}(X^{(s)}).
$$
The direct coupling between the systems enters the QSDEs (\ref{dXs}) only through the terms  $i[H_{12}, X^{(s)}]$ which are computed as follows. A combination of (\ref{H12}),  the antisymmetry  and derivation properties of the commutator with the cross-commutativity (\ref{X12comm}) yields
\begin{align}
\nonumber
     i[H_{12}, X^{(1)}]
     & =
     -i[X^{(1)}, H_{12}]\\
\nonumber
     & =
     -i[X^{(1)}, X^{(1)\rT}] E^{(12)} X^{(2)}\\
\nonumber
     & =
     2(\Theta^{(1)}\cdot X^{(1)})E^{(12)} X^{(2)}\\
\label{iH12X1}
     & =
    2
    \Delta_1
    (X^{(1)}\ox (E^{(12)}X^{(2)}))
    =
     F_1 X^{(12)},
\end{align}
with $X^{(12)}$,  $F_1$ given by (\ref{X12X21}), (\ref{F1}),
where use is also made of (\ref{XCCRTheta}) (applied to the first constituent system), the second equality from (\ref{cdotdiam}) and the sections of the CCR array $\Theta^{(1)}$ in (\ref{CCRs}) assembled into the matrix $\Delta_1$ in (\ref{Deltas}). In a similar fashion,
\begin{equation}
\label{iH12X2}
     i[H_{12}, X^{(2)}]
     =
    2
    \Delta_2
    (I_{n_2}\ox E^{(21)})
    (X^{(2)}\ox X^{(1)})=
     F_2
    X^{(12)},
\end{equation}
where (\ref{X12X21}) is also used,  and $F_2$ is given by (\ref{F2}).
Substitution of (\ref{iH12X1}), (\ref{iH12X2}) into (\ref{dXs}) leads to the QSDEs
\begin{equation}
\label{dXs1}
  \rd X^{(s)}
  =
    (A^{(s)}X^{(s)} + b^{(s)}+F_s X^{(12)}) \rd t + B^{(s)}(X^{(s)})\rd W^{(s)},
    \qquad
    s = 1,2.
\end{equation}
Their right-hand sides
involve the vector $X^{(12)}$  from (\ref{X12X21}) which is a bilinear function of $X^{(1)}$, $X^{(2)}$. Hence, a combination of (\ref{iH12X1})--(\ref{dXs1}) with the quantum Ito lemma yields the QSDE
\begin{align}
\nonumber
    \rd X^{(12)}
    & =
    (\rd X^{(1)})\ox X^{(2)}
    +
    X^{(1)}\ox \rd X^{(2)}
    +
    \overbrace{\rd X^{(1)}\ox \rd X^{(2)}}^0\\
\nonumber
    =&
    ((A^{(1)}X^{(1)} + b^{(1)}+i[H_{12}, X^{(1)}]) \rd t + B^{(1)}(X^{(1)})\rd W^{(1)})\ox X^{(2)}\\
\nonumber
    & +
    X^{(1)}\ox
    ((A^{(2)}X^{(2)} + b^{(2)}+i[H_{12}, X^{(2)}]) \rd t + B^{(2)}(X^{(2)})\rd W^{(2)})\\
\nonumber
    =&
    (
    \sP (b^{(2)}\ox I_{n_1}) X^{(1)}
    +
    (b^{(1)}\ox I_{n_2}) X^{(2)}
    +
    (A^{(1)} \op A^{(2)})X^{(12)} +
    i[H_{12}, X^{(12)}]) \rd t\\
\label{dX12}
    & +
    (B^{(1)}(X^{(1)})\rd W^{(1)})\ox X^{(2)}
    +
    X^{(1)}\ox
    (B^{(2)}(X^{(2)})\rd W^{(2)}),
\end{align}
where use is made of the identity $[H_{12}, X^{(1)}]\ox X^{(2)} + X^{(1)}\ox [H_{12}, X^{(2)}] = [H_{12}, X^{(12)}]$. Here, the Ito correction term vanishes:
$$
    \rd X^{(1)}\ox \rd X^{(2)}
    =
    (B^{(1)}(X^{(1)})\ox B^{(2)}(X^{(2)}))
    (\rd W^{(1)}\ox \rd W^{(2)}) = 0,
$$
since   the future-pointing Ito increments $\rd W^{(s)}$ commute with adapted quantum processes (taken at the same or an earlier moment of time),  and
$\rd W^{(1)}\rd W^{(2)\rT} = 0$ in view of the block-diagonal structure  of the quantum Ito matrix $\Omega$ in (\ref{wwww}),  (\ref{Om12}). By using (\ref{X12CCR}), (\ref{H12}), it follows that
\begin{align}
\nonumber
    i[H_{12}, X^{(12)}]
     = &
     -i[X^{(12)}, H_{12}]\\
\nonumber
    = &
    -i[X^{(12)}, X^{(12)\rT}]
    \vec{E}^{(21)}    \\
\nonumber
    = &
    2
    \Big(
    \sum_{j=1}^{n_1}
    (\Theta_j^{(1)} \ox \alpha^{(2)})
    X_j^{(1)}
    +
    \sum_{k=1}^{n_2}
    (\alpha^{(1)}
    \ox
    \Theta_k^{(2)})
    X_k^{(2)}\\
\nonumber
    & +
    \sum_{j=1}^{n_1}
    \sum_{k=1}^{n_2}
    (\Theta_j^{(1)}\ox  \Re \beta_k^{(2)} +  \Re \beta_j^{(1)}\ox \Theta_k^{(2)})
    \Xi_{jk}^{(12)}
    \Big) \vec{E}^{(21)}     \\
\label{XFX}
    = &
    G_1 X^{(1)} + G_2 X^{(2)} + G_{12} X^{(12)}.
\end{align}
where the matrices $G_s \in \mR^{n_1n_2\x n_s}$, $G_{12}\in \mR^{n_1n_2\x n_1n_2}$ are given by
(\ref{G1})--(\ref{Rebet}).
Also, by using the linear dependence of the matrix $B^{(s)}(X^{(s)})$ in (\ref{BXs}) on $X^{(s)}$,   the diffusion terms in (\ref{dX12}) can be represented as
\begin{align}
\nonumber
        (B^{(1)}(X^{(1)})\rd W^{(1)})\ox X^{(2)}
        & =
        2((\Theta^{(1)}\cdot X^{(1)}) M^{(1)\rT} \rd W^{(1)})\ox X^{(2)}\\
\nonumber
        & =
        2((\Theta^{(1)}\cdot X^{(1)}) \ox X^{(2)})  M^{(1)\rT} \rd W^{(1)}\\
\nonumber
        & =
        2 \sum_{\ell=1}^{n_1}
        (\Theta_\ell^{(1)} \ox (X_\ell^{(1)} X^{(2)}))
        M^{(1)\rT} \rd W^{(1)}\\
\nonumber
        & =
        2
        \sum_{k,\ell=1}^{n_1}
        (\theta_{\bullet k \ell}^{(1)}\ox (X_\ell^{(1)} X^{(2)}))
        (M^{(1)\rT} \rd W^{(1)})_k\\
\nonumber
        & =
        2
        \sum_{k,\ell=1}^{n_1}
        (\theta_{\bullet k \ell}^{(1)}\ox I_{n_2}) X_\ell^{(1)} X^{(2)}
        (M^{(1)\rT} \rd W^{(1)})_k\\
\nonumber
        & =
        2
        \sum_{k=1}^{n_1}
        (\theta_{\bullet k \bullet}^{(1)} \ox I_{n_2})
        X^{(12)}
        (M^{(1)\rT} \rd W^{(1)})_k\\
\label{fB1}
        & = \fB_1(X^{(12)})\rd W^{(1)},\\
\nonumber
    X^{(1)}\ox
    (B^{(2)}(X^{(2)})\rd W^{(2)})
    & =
    2
    X^{(1)}\ox
    ((\Theta^{(2)}\cdot X^{(2)})
    M^{(2)\rT }\rd W^{(2)})\\
\nonumber
    & =
    2
    (X^{(1)}\ox
    (\Theta^{(2)}\cdot X^{(2)}))
    M^{(2)\rT }\rd W^{(2)}\\
\nonumber
    & =
    2
    \sum_{\ell=1}^{n_2}
    (X^{(1)}\ox
    (\Theta_\ell^{(2)}X_\ell^{(2)}))
    M^{(2)\rT }\rd W^{(2)}\\
\nonumber
    & =
    2
    \sum_{k=1}^{n_2}
    (X^{(1)}\ox
    (\theta_{\bullet k \bullet}^{(2)}X^{(2)}))
    (M^{(2)\rT }\rd W^{(2)})_k\\
\nonumber
    & =
    2
    \sum_{k=1}^{n_2}
    (I_{n_1}\ox
    \theta_{\bullet k \bullet}^{(2)})
    X^{(12)}
    (M^{(2)\rT }\rd W^{(2)})_k\\
\label{fB2}
    & =
    \fB_2(X^{(12)})\rd W^{(2)},
\end{align}
where $\fB_s(X^{(12)})$ is an $(n_1n_2\x m_s)$-matrix of self-adjoint operators which are linear functions of $X^{(12)}$ given by (\ref{fB1def}), (\ref{fB2def}).
By substituting (\ref{XFX})--(\ref{fB2}) into (\ref{dX12}), this  QSDE takes the form
\begin{align}
\nonumber
    \rd X^{(12)}
    =&
    (
    (\sP (b^{(2)}\ox I_{n_1}) + G_1) X^{(1)}\\
\nonumber
    & +
    (b^{(1)}\ox I_{n_2} + G_2) X^{(2)}\\
\nonumber
    & +
    (A^{(1)} \op A^{(2)} + G_{12})X^{(12)})\rd t \\
\label{dX12new}
    & +
    \sum_{s=1}^2
    \fB_s(X^{(12)})\rd W^{(s)}.
\end{align}
The QSDEs (\ref{dXs1}), (\ref{dX12new}) can now be assembled into the QSDE (\ref{dX1}) for the vector $X$ in (\ref{Xvec}) with (\ref{Acomp})--(\ref{Bcomp}).
\end{proof}

We will now apply the weak-coupling framework of Sections~\ref{sec:asy}, \ref{sec:inv} to the system interconnection described above. To this end, suppose the coupling of the constituent systems   to the external fields in (\ref{LMNs}), (\ref{LMN}) involves a coupling strength parameter $\eps \> 0$ as
\begin{equation}
\label{Mepsclos}
    M_\eps^{(s)} := \eps \sM^{(s)},
    \qquad
    N_\eps^{(s)} := \eps \sN^{(s)},
    \qquad
    s = 1,2,
\end{equation}
where $\sM^{(s)} \in \mR^{m_s\x n_s}$, $\sN^{(s)} \in \mR^{m_s}$  specify the coupling shape, so that (\ref{MNeps}) holds for the system interconnection with
\begin{equation}
\label{MNcomp}
    \sM:=
    \begin{bmatrix}
      \sM^{(1)} & 0 & 0\\
      0 & \sM^{(2)} & 0
    \end{bmatrix},
    \qquad
    \sN
    :=
    \begin{bmatrix}
      \sN^{(1)}\\
      \sN^{(2)}
    \end{bmatrix}    .
\end{equation}
The corresponding $\eps$-dependent matrix $A_\eps$ in (\ref{Acomp}) and the vector $b_\eps$  in (\ref{bcomp}) can be represented as in (\ref{ABeps}), (\ref{beps}):
\begin{equation}
\label{Abepscomp}
    A_\eps
     = A_0 + \eps^2 \sA,
     \qquad
     b_\eps = \eps^2 \sb,
\end{equation}
where
\begin{equation}
\label{sAsb}
    \sA
    =
        \begin{bmatrix}
      \sA^{(1)} & 0 & 0\\
      0 & \sA^{(2)} & 0\\
    \sP (\sb^{(2)}\ox I_{n_1})
    &
    \sb^{(1)}\ox I_{n_2}
    &
    \sA^{(1)} \op \sA^{(2)}
    \end{bmatrix},
    \qquad
    \sb
    =
    \begin{bmatrix}
      \sb^{(1)}\\
      \sb^{(2)}\\
      0
    \end{bmatrix}.
\end{equation}
Here,
\begin{align}
\label{At1s}
    \sA^{(s)}
    & =
        2
        \Theta^{(s)} \diam (\sM^{(s)\rT} J_s\sN^{(s)})
    +
    2
    \sum_{\ell = 1}^{n_s}
    \Theta_\ell^{(s)}
    \sM^{(s)\rT}
    (
        \sM^{(s)}\theta_{\ell\bullet \bullet}^{(s)}
        +
        J_s \sM^{(s)}\Re \beta_{\ell\bullet \bullet}^{(s)}
    ),\\
\label{bepscomp}
     \sb^{(s)}
     & :=
    2
    \sum_{\ell = 1}^{n_s}
    \Theta_\ell^{(s)}
    \sM^{(s)\rT}
    J_s\sM^{(s)}\alpha_{\bullet \ell}^{(s)},
    \qquad
    s = 1,2,
\end{align}
are independent of the energy vectior $E$ (which parameterises the direct coupling through the Hamiltonian  (\ref{HEX})) but
depend on the system-field coupling shape parameters in (\ref{Mepsclos}), (\ref{MNcomp}) and
result from applying (\ref{At1}), (\ref{beps}) to the constituent systems. 
The matrix
\begin{equation}
\label{A0comp}
      A_0
    =
    2 \Theta \diam E
    =
    \begin{bmatrix}
      A_0^{(1)} & 0 & F_1\\
      0 & A_0^{(2)} & F_2\\
    G_1
    &
    G_2
    &
    A_0^{(1)} \op A_0^{(2)} + G_{12}
    \end{bmatrix}
\end{equation}
in (\ref{Abepscomp}) is the limit value of $A$ in  (\ref{Acomp}) when $\eps = 0$, where,
in accordance with (\ref{A0}),
\begin{equation}
\label{A0s}
    A_0^{(s)}
    =
    2 \Theta^{(s)} \diam E^{(s)},
    \qquad
    s = 1,2,
\end{equation}
and $\Theta^{(1)}$, $\Theta^{(2)}$, $\Theta$ are the CCR arrays for the constituent systems and their interconnection.
This limiting case corresponds to the absence of the system-field   coupling, when the systems
in Fig.~\ref{fig:system} are isolated from  the  external fields but are directly  coupled to each other,  with the energy vector $E$ in (\ref{HEX}) being fixed.
Now, in application to the block diagonal matrix $\alpha$ in (\ref{alfcomp}), the condition (\ref{alfpos}) of Theorem~\ref{th:osc} is equivalent to
\begin{equation}
\label{alf12pos}
    \alpha^{(s)} \succ 0 ,
    \qquad
    s = 1,2
\end{equation}
(since the Kronecker product preserves positive definiteness)
and guarantees that for any energy vector $E\in \mR^n$,  the eigenvalues of the matrix $A_0$ in (\ref{A0comp}) are imaginary (possibly zero). The set of the eigenfrequencies  $\omega_1, \ldots, \omega_n \in \mR$  of $A_0$ for the system interconnection does not reduce to $\{\omega_j^{(1)}, \omega_k^{(2)},  \omega_j^{(1)}+ \omega_k^{(2)}:\ j = 1, \ldots, n_1,\ k = 1, \ldots, n_2\}$  in terms of those for the matrices $A_0^{(1)}$, $A_0^{(2)}$. However, it does so when the direct coupling between the constituent systems vanishes (that is, $E^{(12)} = 0$) and hence, so also do the matrices $F_1$, $F_2$, $G_1$, $G_2$, $G_{12}$ in (\ref{F1}), (\ref{F2}), (\ref{G1})--(\ref{G12}), in which case,  (\ref{A0comp}) reduces to
$$
    A_0
    =
    \begin{bmatrix}
      A_0^{(1)} & 0 & 0\\
      0 & A_0^{(2)} & 0\\
    0
    &
    0
    &
    A_0^{(1)} \op A_0^{(2)}
    \end{bmatrix},
$$
specified completely  by the matrices (\ref{A0s}).
For a nonvanishing direct coupling (when $E^{(12)}\ne 0$),  the matrix $E^{(12)}$  can be varied so as to change the eigenfrequencies and the other eigendata (\ref{simroot})--(\ref{AV}) of the matrix $A_0$ in (\ref{A0comp}). Under the condition (\ref{alf12pos}) on the individual structure constants, in the case  (\ref{omdiff}) of pairwise different  eigenfrequencies $\omega_1, \ldots, \omega_n$ of the matrix $A_0$, the asymptotic behaviour of the eigenvalues $\lambda_1(\eps), \ldots, \lambda_n(\eps)$  of the matrix $A_\eps$ in (\ref{Acomp}) is described by (\ref{eigasy}) of Theorem~\ref{th:asy} in terms of the quantities $\nu_1, \ldots, \nu_n$ from (\ref{nu}).  The latter are computed as the diagonal entries of the matrix $V^*\alpha^{-1/2} \sA \sqrt{\alpha} V$, where the matrix $\sA$ is given by  (\ref{sAsb}), (\ref{At1s}),  and
$$
    \sqrt{\alpha}
    =
    \begin{bmatrix}
      \sqrt{\alpha^{(1)}} & 0 & 0\\
      0 & \sqrt{\alpha^{(2)}} & 0\\
      0 & 0 & \sqrt{\alpha^{(1)}}\ox \sqrt{\alpha^{(2)}}
    \end{bmatrix}
$$
inherits the block diagonal structure of (\ref{alfcomp}).  Their real parts $\Re \nu_1, \ldots, \Re\nu_n$ are responsible (through the condition (\ref{nuneg}) of Theorem~\ref{th:stab}) for the stability of the system interconnection for all sufficiently small values of $\eps>0$ and participate in the Lyapunov exponents (\ref{lead}), decoherence time estimates (\ref{tauhat}) and system-field coupling strength thresholds (\ref{epsmax}). Also, the quantities (\ref{nu}) affect the asymptotic behaviour of the invariant state for the system interconnection through the relation (\ref{mulim}) of Theorem~\ref{th:Plim}  (for odd dimensions $n$ in (\ref{nnn}), that is, if at least one of the dimensions $n_1$, $n_2$ is odd), which involves the vector $\sb$ from (\ref{sAsb}), (\ref{bepscomp}). The dependence of these quantities on the direct energy coupling matrix $E^{(12)}$ can be used for a rational choice of this matrix in order to achieve given specifications on the stability margins and decoherence levels for the system interconnection in the weak-coupling framework.

\section{Concluding remarks}\label{sec:conc}

For a class of open quantum stochastic systems,  whose dynamic variables have an algebraic structure,  and the Hamiltonian and coupling operators depend linearly on them, we have discussed a particular way to quantify the dissipative decoherence effects in terms of the exponential decay in the two-point CCRs.  In these  decoherence measures, including the decoherence time constants,  we have employed the first and second moments of the time-ordered operator exponential which relates the system variables at different times. The practical computability of these averaged quantities (at least in the case of external fields in the vacuum state) exploits the quasilinearity of the QSDE which governs the system. The decay rates in the open system in the form of relevant Lyapunov exponents have been considered in comparison with the time scales of the oscillatory modes of the isolated quantum dynamics in the absence of coupling with  the environment. Using matrix spectrum perturbation techniques, we  have obtained asymptotic decoherence estimates in a weak-coupling formulation involving a small coupling strength parameter along with a given coupling shape.
The asymptotic behaviour of the
invariant quantum state of
the system in the weak-coupling limit has also been discussed.
These results have been illustrated for finite-level quantum systems (and their interconnection through a direct energy coupling)  with multichannel external fields and the Pauli matrices as system variables. The findings of the paper can be of use for performance criteria and optimization in the context of quantum information processing where controlled isolation and decoherence issues  play an important role for system interconnections.


\begin{thebibliography}{99}

\bibitem{B_1983}
V.P.Belavkin, On the theory of controlling observable quantum
systems, \emph{Autom. Rem. Contr.}, vol. 44, no. 2, 1983, pp. 178--188.



\bibitem{BP_2006}
H.-P.Breuer,  and F.Petruccione,
\textit{The Theory of Open Quantum Systems},
Clarendon Press,
Oxford, 2006.

\bibitem{CL_1985}
A.O.Caldeira, and A.J.Leggett,
Influence of damping on quantum interference: An exactly soluble model,
\textit{Phys. Rev. A}, vol. 31, no. 2, 1985,  pp. 1059--1066.


\bibitem{CH_1971}
C.D.Cushen, and R.L.Hudson, A quantum-mechanical central limit theorem,
\emph{J. Appl. Prob.}, vol. 8, no. 3, 1971, pp. 454--469.

\bibitem{DP_2010}
D.Dong, and I.R.Petersen, Quantum control theory and applications: a survey,
\emph{IET Contr. Theor. Appl.}, vol. 4, no. 12, 2010, pp. 2651--2671.


\bibitem{EMPUJ_2016}
L.A.D.Espinosa, Z.Miao, I.R.Petersen, V.Ugrinovskii, and
M.R.James,
Physical realizability and preservation of
commutation and anticommutation relations for
$n$-level quantum systems
\textit{SIAM J. Control Optim.},
vol. 54, no. 2, 2016, pp. 632--661.

\bibitem{GZ_2004}
C.W.Gardiner, and P.Zoller,
\emph{Quantum Noise},
Springer, Berlin, 2004.

\bibitem{GKS_1976}
V.Gorini, A.Kossakowski, E.C.G.Sudarshan, Completely positive dynamical semigroups of $N$-level systems,
\emph{J. Math. Phys.}, vol. 17, no. 5, 1976, pp. 821--825.

\bibitem{D_2006}
M. de Gosson,
\emph{Symplectic Geometry and Quantum Mechanics},
Birkh\"{a}user, Basel, 2006.

\bibitem{H_2008}
N.J.Higham,
\emph{Functions of Matrices}, SIAM, 2008.

\bibitem{H_1996}
A.S.Holevo, Exponential formulae in quantum stochastic calculus,
\emph{Proc. Roy. Soc. Edinburgh}, vol. 126A, 1994, pp. 375--389.

\bibitem{H_2001}
A.S.Holevo, \emph{Statistical Structure of Quantum Theory}, Springer, Berlin, 2001.

\bibitem{HJ_2007}
R.A.Horn, and C.R.Johnson,
\textit{Matrix Analysis},
Cambridge
University Press, New York, 2007.


\bibitem{HP_1984}
R.L.Hudson, and K.R.Parthasarathy,
Quantum Ito's formula and stochastic evolutions,
\emph{Commun. Math. Phys.}, vol. 93,  1984, pp. 301--323.


\bibitem{KS_1991}
I.Karatzas, and S.E.Shreve,
\emph{Brownian Motion and Stochastic Calculus}, 2nd Ed.,
Springer, New York, 1991.

\bibitem{L_1976}
G.Lindblad, On the generators of quantum dynamical semigroups,
\emph{Comm. Math. Phys.}, vol. 48, 1976, pp. 119--130.



\bibitem{M_1985}
J.R.Magnus, On differentiating eigenvalues and eigenvectors,
\textit{Econometric Theory},
vol. 1, no. 2, 1985, pp. 179--191.

\bibitem{M_1988}
J.R.Magnus,
\textit{Linear Structures},
Oxford University Press, New York, 1988.

\bibitem{M_1995}
P.-A.Meyer,
\textit{Quantum Probability for Probabilists},
Springer, Berlin, 1995.

\bibitem{NC_2000}
M.A.Nielsen, and I.L.Chuang,
\textit{Quantum Computation and Quantum Information},
Cambridge University Press, Cambridge, 2000.


\bibitem{NY_2017}
H.I.Nurdin, and N.Yamamoto,
\textit{Linear Dynamical Quantum Systems},
Springer, Netherlands, 2017.

\bibitem{P_1992}
K.R.Parthasarathy,
\emph{An Introduction to Quantum Stochastic Calculus},
Birk\-h\"{a}user, Basel, 1992.

\bibitem{KRP_2010}
K.R.Parthasarathy,
What is a Gaussian state?
\emph{Commun. Stoch. Anal.}, vol. 4, no. 2, 2010, pp. 143--160.

\bibitem{PS_1972}
K.R.Parthasarathy, and K.Schmidt,
\emph{Positive Definite Kernels, Continuous Tensor Products, and Central Limit Theorems of Probability Theory},
Springer-Verlag, Berlin, 1972.

\bibitem{P_2017}
I.R.Petersen,
Quantum linear systems theory,
\textit{Open Automat. Contr. Syst. J.},
vol. 8, 2017, pp. 67--93.

\bibitem{S_1994}
J.J.Sakurai,
\emph{Modern Quantum Mechanics},
 Addison-Wesley, Reading, Mass., 1994.


\bibitem{SIG_1998}
R.E.Skelton, T.Iwasaki, and K.M.Grigoriadis,
\textit{A Unified Algebraic Approach to Linear Control Design},
Taylor \& Francis, London, 1998.

\bibitem{U_1995}
W.G.Unruh,
Maintaining coherence in quantum computers,
\textit{Phys. Rev. A}, vol. 51, no. 2, 1995, pp. 992--997.




\bibitem{VPJ_2018a}
I.G.Vladimirov, I.R.Petersen, and M.R.James, Multi-point Gaussian states, quadratic–expo\-nential cost functionals, and large deviations estimates for linear quantum stochastic systems, \textit{Appl. Math. Optim.}, vol. 83, 2021, pp. 83--137 (published online 24 July 2018).


\bibitem{VP_2022}
I.G.Vladimirov, and I.R.Petersen,
Moment dynamics and observer design for a class of quasilinear quantum stochastic systems,
\textit{SIAM J. Control Optim.},
vol. 60, no. 3, pp.
1223--1249.

\bibitem{VP_2022_decoh}
I.G.Vladimirov, and I.R.Petersen,
Decoherence quantification through commutation relations decay for open quantum harmonic oscillators, submitted (preprint: 	arXiv:2208.02534 [quant-ph], 4 August 2022).



\bibitem{W_1967}
W.M.Wonham, Optimal stationary control of a linear system with
state-dependent noise, \textit{SIAM J. Control}, vol. 5, no. 3, 1967, pp. 486--500.

\bibitem{Y_2012}
N.Yamamoto, Pure Gaussian state generation via dissipation: a quantum stochastic differential
equation approach, \textit{Philos. Trans. R. Soc. Lond. Ser. A, Math. Phys. Eng. Sci.}, vol. 370, 2012, pp. 5324--5337.

\bibitem{Y_2009}
M.Yanagisawa, Non-Gaussian state generation from linear elements via feedback,
\textit{Phys. Rev. Lett.}, vol. 103, no. 20, pp. 203601-1--4.

\bibitem{ZJ_2011a}
G.Zhang, and M.R.James,  Direct and indirect couplings in coherent feedback control of linear quantum systems,
\textit{IEEE Trans. Automat. Contr.}, vol. 56, no. 7, 2011, 1535--1550.

\end{thebibliography}
\end{document}